 \documentclass[journal]{IEEEtran}

\usepackage[bookmarks,colorlinks,pagebackref]{hyperref}
\usepackage{amsthm}
\usepackage{amsmath,balance}
\usepackage{graphicx}
\usepackage{color, verbatim}
\usepackage[usenames,dvipsnames]{xcolor}
\usepackage{mdwmath}
\usepackage{amsmath}
\usepackage{epsfig}
\usepackage{url}
\usepackage{tikz}
\usetikzlibrary{petri}
\usetikzlibrary{fit}
\usetikzlibrary{arrows,matrix,positioning}
\usepgflibrary{shapes.multipart} 
\usepackage{amssymb}
\usepackage{cleveref}
\usepackage{subfig}
\usepackage{multirow}
\usepackage{array}
\usepackage{pgfplots}
\pgfplotsset{compat}
\usepackage{mathtools}
\usepackage{cite,calc}

\usepackage{enumitem}
\usepackage[norelsize,linesnumbered]{algorithm2e}

\DeclarePairedDelimiter{\ceil}{\lceil}{\rceil}
\DeclarePairedDelimiter{\floor}{\lfloor}{\rfloor}
\DeclarePairedDelimiter{\set}{\lbrace}{\rbrace}
\DeclarePairedDelimiter{\paren}{\lparen}{\rparen}
\DeclarePairedDelimiter{\br}{\lparen}{\rparen}
\DeclarePairedDelimiter{\brac}{\lbrack}{\rbrack}
\DeclarePairedDelimiter{\abs}{\lvert}{\rvert}

\crefname{property}{property}{properties}
\crefname{condition}{condition}{conditions}
\crefname{Appendix}{Appendix}{Appendices}
\crefname{algocf}{algorithm}{algorithms}
\Crefname{algocf}{Algorithm}{Algorithms}
\crefname{appsec}{Appendix}{Appendices}

\def\defeq{\mathrel{\mathop:}=}

\newcommand\nc\newcommand
\nc\bfa{{\boldsymbol a}}\nc\bfA{{\boldsymbol A}}\nc\cA{{\mathcal A}}\nc\frA{{\mathfrak A}}\nc\fra{{\mathfrak a}}\nc\bbA{{\mathbb A}}\nc\bba{{\mathbb a}}
\nc\bfb{{\boldsymbol b}}\nc\bfB{{\boldsymbol B}}\nc\cB{{\mathcal B}}\nc\frB{{\mathfrak B}}\nc\frb{{\mathfrak b}}\nc\bbB{{\mathbb B}}\nc\bbb{{\mathbb b}}
\nc\bfc{{\boldsymbol c}}\nc\bfC{{\boldsymbol C}}\nc\cC{{\mathcal C}}\nc\frC{{\mathfrak C}}\nc\frc{{\mathfrak c}}\nc\bbC{{\mathbb C}}\nc\bbc{{\mathbb c}}
\nc\bfd{{\boldsymbol d}}\nc\bfD{{\boldsymbol D}}\nc\cD{{\mathcal D}}\nc\frD{{\mathfrak D}}\nc\frd{{\mathfrak d}}\nc\bbD{{\mathbb D}}\nc\bbd{{\mathbb d}}
\nc\bfe{{\boldsymbol e}}\nc\bfE{{\boldsymbol E}}\nc\cE{{\mathcal E}}\nc\frE{{\mathfrak E}}\nc\fre{{\mathfrak e}}\nc\bbE{{\mathbb E}}\nc\bbe{{\mathbb e}}
\nc\bff{{\boldsymbol f}}\nc\bfF{{\boldsymbol F}}\nc\cF{{\mathcal F}}\nc\frF{{\mathfrak F}}\nc\frf{{\mathfrak f}}\nc\bbF{{\mathbb F}}\nc\bbf{{\mathbb f}}
\nc\bfg{{\boldsymbol g}}\nc\bfG{{\boldsymbol G}}\nc\cG{{\mathcal G}}\nc\frG{{\mathfrak G}}\nc\frg{{\mathfrak g}}\nc\bbG{{\mathbb G}}\nc\bbg{{\mathbb g}}
\nc\bfh{{\boldsymbol h}}\nc\bfH{{\boldsymbol H}}\nc\cH{{\mathcal H}}\nc\frH{{\mathfrak H}}\nc\frh{{\mathfrak h}}\nc\bbH{{\mathbb H}}\nc\bbh{{\mathbb h}}
\nc\bfi{{\boldsymbol i}}\nc\bfI{{\boldsymbol I}}\nc\cI{{\mathcal I}}\nc\frI{{\mathfrak I}}\nc\fri{{\mathfrak i}}\nc\bbI{{\mathbb I}}\nc\bbi{{\mathbb i}}
\nc\bfj{{\boldsymbol j}}\nc\bfJ{{\boldsymbol J}}\nc\cJ{{\mathcal J}}\nc\frJ{{\mathfrak J}}\nc\frj{{\mathfrak j}}\nc\bbJ{{\mathbb J}}\nc\bbj{{\mathbb j}}
\nc\bfk{{\boldsymbol k}}\nc\bfK{{\boldsymbol K}}\nc\cK{{\mathcal K}}\nc\frK{{\mathfrak K}}\nc\frk{{\mathfrak k}}\nc\bbK{{\mathbb K}}\nc\bbk{{\mathbb k}}
\nc\bfl{{\boldsymbol l}}\nc\bfL{{\boldsymbol L}}\nc\cL{{\mathcal L}}\nc\frL{{\mathfrak L}}\nc\frl{{\mathfrak l}}\nc\bbL{{\mathbb L}}\nc\bbl{{\mathbb l}}
\nc\bfm{{\boldsymbol m}}\nc\bfM{{\boldsymbol M}}\nc\cM{{\mathcal M}}\nc\frM{{\mathfrak M}}\nc\frm{{\mathfrak m}}\nc\bbM{{\mathbb M}}\nc\bbm{{\mathbb m}}
\nc\bfn{{\boldsymbol n}}\nc\bfN{{\boldsymbol N}}\nc\cN{{\mathcal N}}\nc\frN{{\mathfrak N}}\nc\frn{{\mathfrak n}}\nc\bbN{{\mathbb N}}\nc\bbn{{\mathbb n}}
\nc\bfo{{\boldsymbol o}}\nc\bfO{{\boldsymbol O}}\nc\cO{{\mathcal O}}\nc\frO{{\mathfrak O}}\nc\fro{{\mathfrak o}}\nc\bbO{{\mathbb O}}\nc\bbo{{\mathbb o}}
\nc\bfp{{\boldsymbol p}}\nc\bfP{{\boldsymbol P}}\nc\cP{{\mathcal P}}\nc\frP{{\mathfrak P}}\nc\frp{{\mathfrak p}}\nc\bbP{{\mathbb P}}\nc\bbp{{\mathbb p}}
\nc\bfq{{\boldsymbol q}}\nc\bfQ{{\boldsymbol Q}}\nc\cQ{{\mathcal Q}}\nc\frQ{{\mathfrak Q}}\nc\frq{{\mathfrak q}}\nc\bbQ{{\mathbb Q}}\nc\bbq{{\mathbb q}}
\nc\bfr{{\boldsymbol r}}\nc\bfR{{\boldsymbol R}}\nc\cR{{\mathcal R}}\nc\frR{{\mathfrak R}}\nc\frr{{\mathfrak r}}\nc\bbR{{\mathbb R}}\nc\bbr{{\mathbb r}}
\nc\bfs{{\boldsymbol s}}\nc\bfS{{\boldsymbol S}}\nc\cS{{\mathcal S}}\nc\frS{{\mathfrak S}}\nc\frs{{\mathfrak s}}\nc\bbS{{\mathbb S}}\nc\bbs{{\mathbb s}}
\nc\bft{{\boldsymbol t}}\nc\bfT{{\boldsymbol T}}\nc\cT{{\mathcal T}}\nc\frT{{\mathfrak T}}\nc\frt{{\mathfrak t}}\nc\bbT{{\mathbb T}}\nc\bbt{{\mathbb t}}
\nc\bfu{{\boldsymbol u}}\nc\bfU{{\boldsymbol U}}\nc\cU{{\mathcal U}}\nc\frU{{\mathfrak U}}\nc\fru{{\mathfrak u}}\nc\bbU{{\mathbb U}}\nc\bbu{{\mathbb u}}
\nc\bfv{{\boldsymbol v}}\nc\bfV{{\boldsymbol V}}\nc\cV{{\mathcal V}}\nc\frV{{\mathfrak V}}\nc\frv{{\mathfrak v}}\nc\bbV{{\mathbb V}}\nc\bbv{{\mathbb v}}
\nc\bfw{{\boldsymbol w}}\nc\bfW{{\boldsymbol W}}\nc\cW{{\mathcal W}}\nc\frW{{\mathfrak W}}\nc\frw{{\mathfrak w}}\nc\bbW{{\mathbb W}}\nc\bbw{{\mathbb w}}
\nc\bfx{{\boldsymbol x}}\nc\bfX{{\boldsymbol X}}\nc\cX{{\mathcal X}}\nc\frX{{\mathfrak X}}\nc\frx{{\mathfrak x}}\nc\bbX{{\mathbb X}}\nc\bbx{{\mathbb x}}
\nc\bfy{{\boldsymbol y}}\nc\bfY{{\boldsymbol Y}}\nc\cY{{\mathcal Y}}\nc\frY{{\mathfrak Y}}\nc\fry{{\mathfrak y}}\nc\bbY{{\mathbb Y}}\nc\bby{{\mathbb y}}
\nc\bfz{{\boldsymbol z}}\nc\bfZ{{\boldsymbol Z}}\nc\cZ{{\mathcal Z}}\nc\frZ{{\mathfrak Z}}\nc\frz{{\mathfrak z}}\nc\bbZ{{\mathbb Z}}\nc\bbz{{\mathbb z}}

\DeclareMathOperator{\supp}{supp}
\DeclareMathOperator{\rank}{rank}

\newcommand{\remove}[1]{}

\newcommand{\colspan}{\mathrm{colspan}}

\newtheorem{theorem}{Theorem}
\newtheorem{definition}{Definition}
\newtheorem{lemma}[theorem]{Lemma}
\newtheorem{proposition}[theorem]{Proposition}
\newtheorem{corollary}[theorem]{Corollary}
\newtheorem{claim}[theorem]{\indent Claim}

\newcommand\ff{{\mathbb F}}

\newcommand\integers{{\mathbb Z}}
\newcommand\rationals{{\mathbb Q}}

\begin{document}
\sloppy

\title{Security in Locally Repairable Storage}
\author{Abhishek Agarwal  and Arya Mazumdar~\IEEEmembership{Senior Member,~IEEE}
\thanks{Abhishek Agarwal is with the Department of Electrical and Computer Engineering, University of Minnesota, Minneapolis, MN  55455, email: \texttt{abhiag@umn.edu.}}
\thanks{Arya Mazumdar is with the College of Information and  Computer Science, University of Massachusetts, Amherst, MA  01003, email: \texttt{arya@cs.umass.edu}. Part of this work has been done when the author was at University of Minnesota.}
 \thanks{This work was supported in part by NSF CCF 1318093, CCF 1453121, CCF 1642658. A preliminary version of this work was presented  in the  IEEE Information Theory Workshop, Jerusalem, Israel, 2015.}}
%
\allowdisplaybreaks
\maketitle

\begin{abstract}
In this paper we extend the notion of {\em locally repairable} codes to {\em secret sharing} schemes. 
The  main problem that we consider is to find  optimal ways to 
distribute shares of a secret 
among a set of storage-nodes (participants) such that the content of each node (share) 
can be recovered by using contents of only few other nodes, and at the same time 
the secret can be reconstructed by only some allowable subsets of nodes.
As a special case, an   eavesdropper observing 
some set of specific nodes (such as less than certain number of nodes) does not get any information. 
In other words, we propose to study
a locally repairable distributed storage system that is secure against a {\em passive eavesdropper} that can observe some subsets of
nodes. 

We provide a number of results related to such systems including upper-bounds and achievability results on the number of bits that can be securely stored with these constraints. 
In particular, we provide conditions under which a locally repairable code can be turned into a  secret sharing scheme and extend the results of
secure repairable storage to cooperative repair and storage on networks.
Additionally, we  consider perfect secret sharing schemes over general access structures under locality constraints and   give an example of a perfect secret sharing scheme
that can have small locality. Lastly,
 we provide a lower bound on the size of a share compared to the size of the secret that shows how locality affects the sizes of shares in a perfect scheme.

\end{abstract}


\section{Introduction}



Secret sharing schemes were proposed by Shamir and Blakley \cite{shamir1979share,blakley1899safeguarding} to provide security against an eavesdropper with unbounded computational capability. Consider the secret  as a realization of a (uniform) random vector $\bfS$ over some support. Define $[n] \defeq \set{1,2,\ldots,n}$ and let $2^A$ denote the power set for set $A$. Suppose that shares of the secret are to be distributed  among $n$ participants (storage nodes) such that a set of shares belonging to $\cA_s \subseteq 2^{[n]}$, is able to determine the secret. $\cA_s$ is called the access structure of the secret sharing scheme. Denote the random variable corresponding to the share of a participant (or node) $i\in [n]$ by $C_i$ and let $\bfC = (C_1 C_2 \ldots C_n)$. Let $\bfx_A$ denote the projection of the vector $\bfx \in \ff^n$ to the co-ordinates in $A\subseteq [n]$. For a singleton set $A =\set{i}$ let $\bfx_i \defeq \bfx_{\set{i}}$. A secure scheme has the property that a subset of shares in the block-list $\cB_s \subseteq 2^{[n]}$ are unable to determine anything about the secret. Thus, $H(\bfS|\bfC_B)=H(\bfS)$ for any $B\in \cB_s$ and $H(\bfS|\bfC_A)=0$ for any $A\in \cA_s$, where $H(\cdot)$ denotes the entropy\footnote{The unit of entropy in this paper is $q$-ary, where $q$ is an integer that will be clear from context.}. For a standard {\em monotone} secret sharing scheme the classes $\cA_s$ and $\cB_s$ must have the following properties,
\begin{gather*}
A^\prime \supseteq A, A \in \cA_s \implies A^\prime \in \cA_s \\
B^\prime \subseteq B, B \in \cB_s \implies B^\prime \in \cB_s \\
\mbox{ and } \\
\cB_s \subseteq 2^{[n]}\setminus\cA_s .
\end{gather*}
For a {\em perfect} secret sharing scheme we have the above monotone property and $\cB_s = 2^{[n]}\setminus \cA_s$. 
Perfect schemes for access structures of the form $\cA_s = \set{A \subseteq [n] : \abs{A}\geq m }$ are called {\em  threshold}  secret sharing schemes.
We refer to \cite{Beimel_secret_sharing_schemes} for a comprehensive survey of secret sharing schemes.


A convenient property of schemes that need to store data in a distributed storage system is local repairability \cite{gopalan2012locality} i.e. any storage node can be repaired by accessing a small subset of other nodes, much smaller than is required for decoding the complete data. Error-correcting codes with the local repair property -- locally repairable codes (LRC) -- have been the center of a lot of research activities lately \cite{gopalan2012locality,papailiopoulos2012locally,barg2013family, cadambe2013upper}. Consider  an $n$ length code over a $q$-ary alphabet, $\cC \subseteq \ff_q^n$ of size $|\cC| =q^k$. The code is said to have {\em locality} $r$, if for every $i$, $1 \leq i \leq n,$  there exists a set $\cR_i \subseteq [n]\setminus\{i\}$ with $|\cR_i|\le r$ such that for any two codewords $\bfu,\bfu^\prime \in \cC$ satisfying  $\bfu_i \ne \bfu^\prime_i,$ we have $\bfu_{\cR_i} \neq \bfu^\prime_{\cR_i}$. In a code with locality $r$, any symbol of a codeword can be deduced by reading only at most $r$ other symbols of the codeword. 
For application in distributed storage, the code is further required to have a large {\em minimum distance}  $d$, since that helps recovery in the event of a catastrophic failures (i.e., up to $d-1$ node failures). It is known that \cite{gopalan2012locality} for such a code,
  \begin{equation}\label{eq:locality}
  d \le n - k -\ceil{k/r} +2,
  \end{equation}
  which is also achievable \cite{papailiopoulos2012locally,barg2013family}. A $q$-ary code of length $n$, size $q^k$ and locality $r$ will be called an $(n,k,r)_q$-optimal LRC if it's minimum distance satisfies \eqref{eq:locality} with equality.
  
Security in distributed storage has recently been considered in a number of papers, for example \cite{shah2011information,pawar2011securing,tandonnew, goparaju2013data} and references therein. In these papers the main objective is to secure stored  or downloaded data against an adversary.  Threshold secret sharing protocols over a network under some communication constraint has been considered in \cite{shah2013secure}.
Problems most closely related to this paper perhaps appear in \cite{rawat2012optimal} where a version of threshold secret sharing scheme with locality has been studied.  
Motivated by the above applications in distributed storage, we analyze secret sharing schemes  with different access structures such that shares of each participant/node can be repaired with locality $r$. 

\subsection{Contributions and organization}\label{sec:contributions}

Our contributions in this paper are summarized in the following list.
\begin{enumerate}[leftmargin=*]
\item {\em Distributed storage.} We provide bounds and achievability results for a locally repairable scheme for access structure and block-list, $\cA_s = \set{A \subseteq [n] : \abs{A}\geq m }$ and $\cB_s = \set{B \subseteq [n] : \abs{B}\leq \ell }$, respectively. As evident from \cref{nklmr_secret_sharing}, this access and block structures model a simple distributed storage scenario. We assume that the shares of the secrets are locally recoverable and at the same time an adversary observing up to $\ell$ shares does not get any information.  A more general version of this model that also considers repair bandwidth as a parameter appears in \cite{rawat2012optimal}. In \cref{sec:converse_results} we also address the conditions under which a locally repairable error-correcting code can be converted into a secret sharing scheme with the above access structure. 

{\em Comparison of this part with results of  \cite{rawat2012optimal}:}  
In \cite{rawat2012optimal},  bounds on secrecy capacity for regenerating and locally recoverable  codes have been derived 
using information theoretic inequalities,
and achievability of
 these bounds using schemes that require Gabidulin precoding technique has been shown. 

Our method to prove the converse result is different from that used in \cite{rawat2012optimal}. 
One advantage of our technique for the bound in \cref{sec:converse_results} is that it can be easily applied to cooperative repair (\cref{sec:schemes_for_co_operative_repair}) and repairable codes on graphs (\cref{sec:security_for_repairable_codes_on_graphs}).

We provide a random coding argument using network flow graphs to show the existence of an achievability scheme for the bound, and also adapt the method of \cite{rawat2012optimal} for more general scenarios mentioned above (i.e., cooperative repair and repairable codes on graphs). For these scenarios, we use \cref{lemma1} and Gabidulin precoding to construct transformations to form secure schemes from existing non-secure locally repairable codes.



\item {\em Maximal recoverability.} The Gabidulin precoding described above can  be used to construct optimal codes  but requires an exponentially large (in $n$) alphabet size. 
A simple construction of secret sharing schemes from LRCs is provided in \cref{construction_lin_code}.
We specify in \cref{lemma1} the additional constraints that an optimal LRC would have to satisfy to be able to construct optimal secret sharing schemes in this method. This shows  that to construct an optimal secure scheme with small share size we essentially need a {\em maximally recoverable
code} over small alphabet (see \cref{secure_tamo_barg}). 
\item {\em Perfect secret sharing with small locality.} In \cref{sec:largest_share_size}, we  consider perfect secret sharing schemes over general access structures under locality constraints. While we show that for  threshold secret sharing schemes, there cannot exist any non-trivial local repairability, we  give an example of a perfect secret sharing scheme that can have  small locality. 
\item   {\em Lower bound on the size of shares in terms of the size of the secret.} Furthering the result of \cite{Csirmaz1997} to locally repairable schemes we provide an analogous lower-bound on the size of a share compared to the size of the secret. We further show how locality effects the sizes of shares in a perfect scheme as they relate to the size of the secret. These results are presented in \cref{sec:largest_share_size} (see \cref{thm:size}).
\item   {\em Extension.} We extend the notion of security to cooperative local repair \cite{rawat2014cooperativelocal} where a Distributed Storage System can deal with simultaneous multiple node failures. 
We provide upper-bounds on the secrecy capacity and construct achievable schemes for this scenario in \cref{sec:schemes_for_co_operative_repair}.
\item{\em Extension.} A different and practical generalization for  secret sharing scheme is made in which the Distributed Storage System is represented by a graph $\cG$ such that a node can only connect to its neighbors in $\cG$ for repair. This scenario has been considered in \cref{sec:security_for_repairable_codes_on_graphs}.
\end{enumerate}

\section{A secret-sharing scheme for distributed storage}\label{sec:converse_results}
We start this section by formally defining a secret sharing scheme for a particular, common access structure and block-list: $\cA_s = \set{A \subseteq [n] : \abs{A}\geq m }$ and $\cB_s = \set{B \subseteq [n] : \abs{B}\leq \ell }$. For a code $\cC \subset \ff_q^n$ and set $I \subset [n]$ define $\cC_I \defeq \set{\bfx_I \in \ff_q^{\abs{I}} : \bfx \in \cC}$.

\begin{definition}\label{nklmr_secret_sharing}
An $(n, k, \ell, m, r)_q$-secret sharing scheme consists of a randomized encoder $f$ that maps a uniform secret $\bfS \in \ff_q^k$
randomly to $\bfC  = f(\bfS) \in \ff_q^n$, and must have the following three properties. 
\begin{enumerate}
\item (Recovery) Given any $m$ symbols of $\bfC$, the secret $\bfS$ is completely determined. This guarantees that the
secret is recoverable even with the loss   of any  $n-m$ shares.
\begin{equation}\label{cond1}
H(\bfS| \bfC_I) = 0, \; \forall I \subseteq [n], \abs{I} = m
\end{equation}
\item (Security) Any set of $\ell$ shares  of $\bfC$ does not reveal anything about the secret.
\begin{equation}\label{cond2}
H(\bfS | \bfC_J) = H(\bfS), \; \forall J \subseteq [n], \abs{J} = \ell
\end{equation}
A scheme satisfying this condition is called $\ell$-secure. An eavesdropper that can observe $\ell$ nodes is called an $\ell$-strength eavesdropper.
\item (Locality) For any share, there exist at most $r$ other shares that completely determine this.
For all $i$, there exists $\cR_i\subseteq [n]\setminus \{i\}: |\cR_i| \le r$, such that 
\begin{equation}\label{cond3}
H(\bfC_i | \bfC_{\cR_i}) = 0
\end{equation}
 $\cR_i$ is called the recovery set of share $i$. 
\end{enumerate}
\end{definition}
The maximum amount of secret that can be stored as a function of $n, \ell, m$ and $r$ is called the capacity of the secret sharing scheme and in the following we provide exact characterization of this quantity.
We can define the security condition above in a modified way where the eavesdropper is allowed to see any set $J\subseteq [n]$ of shares and we calculate the amount of information revealed, i.e. $I(\bfS ; \bfC_J)$, in terms of $n, k, |J|, m $ and $r$ in an optimal scheme. This extension is easy from our result and somewhat summarized in \cref{cor:grad}.  


Note that, for locally repairable schemes with no security requirement i.e. $\ell=0$ the following lower-bound on $m$
is apparent from \eqref{eq:locality},
\begin{equation}\label{eq:locality1}
m \ge k +\ceil{k/r} -1,
\end{equation}
This lower bound follows from the definition of the minimum distance of a code $d =n-m+1$. In the subsequent, we provide the fundamental limit on secrecy capacity and constructions achieving that limit.  

As mentioned in the introduction,  a generalized version of this type of secret-sharing scheme that include repair-bandwidth and other  parameters was studied in \cite{rawat2012optimal}. Our \cref{easy_upper_bound,thm:achievability} can be obtained as a consequence of results of that paper. We still provide different proofs of these results as the concepts introduced will be useful for later developments.


\subsection{Bounds}
Let us first prove an immediate and naive upper bound on the capacity of a locally repairable secret sharing scheme that follows as a consequence of 
Eq.~\eqref{eq:locality1}.
\begin{proposition}
For any $(n,k,\ell, m, r)_q$-secret sharing scheme,
$$ k \leq m-\ell - \floor*{\frac{m-\ell}{r+1}}$$
\end{proposition}
\begin{proof}
Consider the randomized encoding $f$ of any $(n,k,\ell, m, r)$-secret sharing scheme. For any secret $\bfs \in \ff_q^k$, define the support of the map $f(\bfs)$ to be  $ \supp(f(\bfs))= \set{\bfx \in \ff_q^n: \Pr(f(\bfs) = \bfx) \ne 0}$. Clearly for any pair $\bfs,\bfs' \in \ff_q^k$ $\bfs\ne \bfs^\prime$, $\supp(f(\bfs)) \cap \supp(f(\bfs')) = \emptyset$.

Suppose, for some $\bfs \in \ff_q^k$, $\bfx \in \supp(f(\bfs))$. Let $I \subseteq [n]$ and $|I| = \ell$. Note that, for each  $\bfs'\in \ff_q^k\setminus \bfs$, there must exist $\bfz \in \supp(f(\bfs'))$ such that $\bfz_I= \bfx_I$ (from the Security property). Let $\cC \subseteq \{\bfz \in \supp(f(\bfs')): \bfs' \in \ff_q^k \text{ and } \bfz_I= \bfx_I\}$ such that $|\cC \cap\supp(f(\bfs'))| =1 \forall \bfs' \in \ff_q^k$. We have $\cC \subseteq \ff_q^n$ and $|\cC| = q^k$. Moreover, from the Recovery property, any $m$ coordinates of a vector in  $\cC$ must be unique, which implies $\cC$ has minimum distance at least $n-m +1$.

Since $\set{f(\bfs):\bfs\in \ff_q^k}$ has locality $r$ any set $\cC\subset \set{f(\bfs):\bfs\in \ff_q^k}$ must have locality $r$. Since, all the codewords in $\cC$ have fixed value on the co-ordinates $I$, $\cC_{[n]\setminus I} \in \ff_q^{n-\ell}$ must be a code of length $n-\ell$ and locality $r$. Moreover, $\cC_{[n]\setminus I}$ has minimum distance at least $n-m +1$ (same as $\cC$). Now from eq.~\eqref{eq:locality} we have,
\begin{subequations}\label{weak_bound}
\begin{align}
&& n-m +1 &\le (n-\ell)-k  - \ceil{k/r} +2 \nonumber\\
&\iff & k + \ceil{k/r} -1 &\le m-\ell \label{weak_bound_f1}\\
&\iff & k &\leq m-\ell - \floor*{\frac{m-\ell}{r+1}} \label{weak_bound_f2}
\end{align}
\end{subequations}
where \cref{weak_bound_f2} follows by replacing both sides of \cref{weak_bound_f1} by $Incr_0(k + \ceil{k/r} -1)$ and $Incr_0(m-\ell)$ respectively,  where $Incr_0(.)$ denotes the increasing function $Incr_0(x) \defeq x-\floor*{\frac{x}{r+1}}$.
\end{proof}

This naive bound in \cref{weak_bound_f1} is not the best possible: it can be further improved to
\begin{equation}\label{eq:trivial}
	k+\ell+\ceil*{\frac{k+\ell}{r}}-1 \leq m.
\end{equation}
To prove \eqref{eq:trivial}, instead of trying to use eq.~\eqref{eq:locality} as a black-box, we follow its proof method \cite{gopalan2012locality,cadambe2013upper}.

\begin{theorem}\label{easy_upper_bound}
Any $(n,k, \ell,m,r)_q$-secret sharing scheme must satisfy,
\begin{equation}\label{easy_bound}
k +\ell \leq m  - \floor*{\frac{m}{r+1}}.
\end{equation}
\end{theorem}
The upper-bound in \cref{easy_bound} can  also be obtained from \cite[Theorem~33]{rawat2012optimal} where the authors use a different method. 
It should be noted that \cref{easy_bound} is equivalent to \cref{eq:trivial}. We see that \cref{eq:trivial} $\implies$ \cref{easy_bound} by replacing both sides in \cref{easy_bound} by the increasing function $Incr_0(x) \defeq x-\floor{x/(r+1)}$. Similarly \cref{easy_bound} $\implies$ \cref{eq:trivial} by replacing each side with the increasing function $Incr_1(x) \defeq x+\ceil{x/r} - 1$. This follows because of the following fact,
{\claim 
For $x,y,r \in \integers^+$,
\begin{equation}\label{eq_rev1}
	y=x + \ceil*{\frac{x}{r}}-1 \iff x = y- \floor*{\frac{y}{r+1}}
\end{equation}}
\begin{IEEEproof}
Let $x=q r + w, \;w <r$. Then, we have.
\begin{subequations}\label{two_forms_UB}
\begin{align}
	& {x + \ceil*{\frac{x}{r}}-1}- \floor*{\frac{{x + \ceil*{\frac{x}{r}}-1}}{r+1}} \label{eq_rev1_2}\\
	&= x + q + \ceil*{\frac{w}{r}} -1 - \floor*{\frac{qr+w+q+\ceil*{\frac{w}{r}}-1}{r+1}}\\
	&= x  +   \ceil*{\frac{w}{r}} - \floor*{\frac{w+\ceil*{\frac{w}{r}}-1}{r+1}} -1 	\\
	&= x\label{eq_rev1_1}
\end{align}
\end{subequations}	
where \cref{eq_rev1_1} follows since $\ceil*{\frac{w}{r}} - \floor*{\frac{w+\ceil*{\frac{w}{r}}-1}{r+1}} -1 =0$ for $w \in [0,r-1]$. Now, substituting $y = x + \ceil*{\frac{x}{r}}-1$ in \cref{eq_rev1_2} we have, \cref{eq_rev1}.
\end{IEEEproof}

\begin{IEEEproof}[Proof of \cref{easy_upper_bound}]
Let $\Lambda_i = \set{\cR_i \cup \set{i}}$. Recall that we can recover the secret $\bfS$ from any $m$ symbols in the $n$ length word $f(\bfS) = \bfC$.  We construct an $m$-subset $\cM \subseteq [n]$ such that  $|\{i: \Lambda_i \subseteq \cM\}|$ is maximized. Suppose, $\displaystyle \cM' = \bigcup_{i:\Lambda_i \subseteq \cM} \cR_i $. 

We have $H(\bfC_\cM | \bfC_{\cM^\prime})=0$. Moreover $H(\bfS|\bfC_\cM)=0$. 
This implies, 
\[H(\bfS|\bfC_{\cM^\prime}) = 0.\]

Now we can select any $\ell$-subset ${\cL}$ of $\cM^\prime$ and assume that the eavesdropper observes that set. Therefore, $H(\bfS) = H(\bfS|\bfC_{\cL})$ must be less than or equal to the number of symbols in $\cM^\prime \setminus \cL$.  Formally,
\begin{align}\label{entropy_manip}
k = H(\bfS) = H(\bfS|\bfC_{\cL}) \leq H(\bfC_{\cM^\prime}|\bfC_{\cL})& \leq \abs{\cM^\prime\setminus \cL}\nonumber\\
& = |\cM'| -\ell.
\end{align}

 This observation will lead   us to \cref{easy_bound}. We describe below, the only remaining task: the method for constructing the set $\cM$ described above, and show that it gives us \cref{easy_bound}. The construction for $\cM$ is given in \cref{algo:conver_result}.

\begin{algorithm}
 \KwData{$\cR_i$ for all $i$}
 \KwResult{$\cM \subseteq [n], \abs{\cM}=m$ containing at least $\floor{m/(r+1)}$ recovery sets }
 $j=0$;
 $\cM^j = \emptyset$\\
 choose any $t \in [n]$\\
 \While{$\abs{\cM^j \cup \set{\Lambda_t}} < m$}{
 $\cM^{j+1} = \cM^j \cup {\Lambda_t}$\\
 choose $t \notin \cM^{j+1}$\\
	$j = j+1$\\
 }
	\eIf{$\abs{\cM^j \cup \Lambda_t} \leq m$}{
		$\cM^{j+1}= \cM^j \cup {\Lambda_t}$\\
	 }{
		$\cI = $ any $(m - \abs{\cM^j})$-subset of $[n] \setminus\cM^j$\\
	 $\cM^{j+1}=\cM^j \cup \cI$
	}
	$j=j+1$\\
	$\cM=\cM^{j}$
 \caption{Constructing a set $\cM\subseteq \{1,2,\dots, n\}$ to maximize $\abs{\set{i: \Lambda_i \subseteq \cM}}$}\label{algo:conver_result}
\end{algorithm}

Note that \cref{algo:conver_result} may not actually give the set containing the maximum number of $\Lambda_i$ but it would suffice to prove the  bound in \cref{easy_bound}. Let $\nu$ denote number of  sets $\Lambda_i$ added to $\cM^0$. We have, $\abs{\Lambda_i} \leq r+1, \forall i$. 
So the maximum size of the set added in each step is $r+1$. Since $\abs{\cM}=m$ by construction, when the algorithm ends at line $9$ we have $\nu \geq \ceil*{\frac{m}{r+1}}$. If the algorithm ends at line $10$ we  must have, $\nu \geq \floor*{\frac{m}{r+1}}$.  
Evidently we have constructed a set $\cM$ such that $|\cM'| = |\cM| - \nu \le m - \floor*{\frac{m}{r+1}}$.
From \cref{entropy_manip} we have,
\begin{align}\label{alg_result}
 k  \leq m  - \floor*{\frac{m}{r+1}} - \ell.
\end{align}
\end{IEEEproof}

Using \cref{entropy_manip} we can show the following,
\begin{corollary} \label{cor:grad}
There exists a set $J \subseteq [n]$ with $\ell \leq \abs{J} \leq m - \floor*{\frac{m}{r+1}}$ such that,
\begin{equation}\label{gradual_degradation}
	H(\bfS|\bfC_{J}) \leq m - \floor*{\frac{m}{r+1}} - \abs{J}.
\end{equation}
\end{corollary}
\Cref{gradual_degradation} gives an upper-bound on the maximum ambiguity of the secret of an $(n,k,\ell,m,r)$-scheme when the eavesdropper has access to more than $\ell$ shares.

\subsection{Constructions}
\label{sec:constr}
It is possible to show matching achievability results to \cref{easy_upper_bound} by
a number of different methods.
\begin{theorem}\label{thm:achievability}
There exists a $(n,k,\ell, m,r)$-secret sharing scheme such that \cref{eq:trivial} is satisfied with equality.
\end{theorem}

In particular this theorem can be proved by constructing a random linear network code. We delegate that proof to \cref{sec:achievability_nklmr}.

The achievability result also follows from  \cite{rawat2012optimal}, that gives a construction for optimal secure LRC  employing Gabidulin codes to satisfy the security constraint. In the subsequent we describe their method, adapted for our scenario, because this will be useful later in our paper when we consider more general secret sharing schemes. 

An intuitive construction of $\ell$-secure schemes comes by replacing some inputs to a LRC with uniform random variables. Formally, consider a linear code $\cC$ with code-length $n$ and dimension  $\br{k+\ell}$. Let $G= [G^1 \;G^2] \in \ff_q^{n\times \br{k+\ell}}$ be the generator matrix of this code such that $G^1 \in \ff_q^{n\times \ell}$ and $G^2 \in \ff_q^{n\times k}$. Let $\bfa \in \ff_q^{k+\ell}$ be the input to the encoder of $\cC$ (i.e., the codeword is generated by multiplying $\bfa$ with the generator matrix of $\cC$). Denote by $\bfs \in \ff_q^{k}$ the input we want to store securely. We construct an $\ell$-secure secret sharing scheme using $\cC$ by taking,
\begin{equation}\label{construction_lin_code}
	\bfa = \begin{bmatrix}
		\bfr \\  \bfs
	\end{bmatrix}
\end{equation}
where $\bfr \in \ff_q^\ell$ is an instance of uniformly distributed random vector. This scheme is $\ell$-secure if and only if for any $\ell$ {\em linearly independent } rows of $G$ the corresponding rows of $G^1$ are linearly independent.

{\lemma \label{lemma1} Let $\bfg_i = [g_{i1} g_{i2} \ldots g_{i(k+\ell)}], i\in [\ell]$ be any  $\ell$ linearly independent rows of $G$. The secret sharing scheme constructed in \cref{construction_lin_code} is $\ell$-secure if and only if the corresponding row vectors $\bfg^1_i = [g_{i1} g_{i2} \ldots g_{i\ell}], i\in [\ell]$ of $G^1$ are linearly independent.}

The proof of \cref{lemma1} is given in \cref{proof_lemma1}. Note that using \cref{lemma1} we can add the security property to any linear code; we do not assume any locality property for the generator matrix $G$. But, it is clear that if the generator matrix $G$ has locality $r$, then so would the scheme constructed in \cref{construction_lin_code}. The construction of an optimal $(n,k,\ell,m,r)_q$ scheme is described in the following. 

\noindent{\em Gabidulin precoding construction:} Let $N$ be an integer.  The points $\alpha_i \in \ff_{q^N}, i\in [n]$ can be represented as vectors in $\ff_q^N$ and are  said to be {\em $\ff_q$-linearly independent} when the corresponding vectors over $\ff_q$ are linearly independent. A Gabidulin code from $\ff_{q^N}^k \rightarrow \ff_{q^N}^n$, for input $(f_1 f_2 \ldots f_k), f_i \in \ff_{q^N}$, is obtained by evaluating the linearized polynomial $\Theta(y) = \sum_{i=1}^{k} f_i {y^q}^{i-1}$ at $n$ $\ff_q$-linearly independent points $\alpha_i \in \ff_{q^N}, i\in [n]$. The linearized polynomial $\Theta(y)$ has the following linearity property,
\begin{equation}\label{gabidulin_linearity}
\Theta(ax+by)=a\Theta(x)+b\Theta(y)
\end{equation}
for all $x,y \in \ff_{q^N}$ and $a,b \in \ff_{q}$. Note that, we need $N\geq n$ to obtain $n$ $\ff_q$-linearly independent points in $\ff_{q^N}$.

Consider the generator matrix, $G = [\bfg_1 \ldots \bfg_n]^T$ of a linear $(n,k+\ell,r)_q$-optimal LRC, where $\bfg_i = [g_{i1} \ldots g_{i(k+\ell)}]^T$. Consider $\bfa = (\bfs\;\; \bfr)$, where $\bfr$ is an instance of uniformly distributed random variable in $\ff_{q^N}^\ell$ and $\bfs\in \ff_{q^N}^k$, $N \ge n$, denotes the secret. First, $\bfa$ is precoded using a Gabidulin code, $\Gamma : \ff_{q^N}^{k+\ell} \rightarrow \ff_{q^N}^{k+\ell} $ which is obtained by evaluating the polynomial,
\begin{equation}\label{Gabidulin_precoding}
\Psi_\bfa(y) = \sum_{i=1}^{k+\ell} a_i {y^q}^{i-1}	
\end{equation}
at the $\ff_q$-linearly independent points $\alpha_i \in \ff_{q^N}, i \in \brac{k+\ell}$. Now, representing $\Gamma(\bfa) \in \ff_{q^N}^{k+\ell} $ as a matrix of size $(k+\ell) \times N$ in $\ff_q$, each column of the matrix can be encoded independently using the generator matrix $G$ for the optimal LRC to get $(c_i)_{i=1}^n = \bfc \in \ff^n_{q^N}$. It is easy to show that this construction is $\ell$-secure. The optimality of the scheme then follows from the optimality of the initial linear LRC. The proof of security of this construction is given below.

\begin{IEEEproof}[Proof of \cref{thm:achievability} with the Gabidulin construction]
Assume without loss of generality (wlog) that the eavesdropper observes $\cE = [\ell] \subseteq [n]$ symbols $c_i, i\in \cE$. Let $\tilde{G} = [\bfg_1 \ldots \bfg_\ell]^T$. Further assume that the $\rank\paren{\tilde{G}}=\ell$, since otherwise the $\ell$-strength eavesdropper is equivalent to an $\rank\paren{\tilde{G}}$-strength eavesdropper. Let $\displaystyle \tilde{\alpha}_i = \sum_{j=1}^{k+\ell} g_{ij} \alpha_j, i\in \cE$. Then since $\tilde{G}$ is full-rank $\set{\tilde{\alpha}_i}_{i\in \cE}$ are $\ff_q$-linearly independent. Therefore, using \cref{gabidulin_linearity} we have,
\begin{align*}
	c_i &= \sum_{j=1}^{k+\ell} g_{ij} \Psi_\bfa(\alpha_j)  \\
				 &= \Psi_\bfa(\sum_{j=1}^{k+\ell} g_{ij} \alpha_j) = \Psi_\bfa(\tilde{\alpha}_i), i\in \cE .
\end{align*}

Let $\bfR, \bfS, \bfC$ be the random variables corresponding to the vector $\bfr$, the secret $\bfs$, and the node shares $\bfC = (C_i)_i$. To prove security we use the secrecy lemma in \cite[Lemma 4]{rawat2012optimal}, to show that $H(\bfC_\cE)\leq H(\bfR)$ and $H(\bfR|\bfS,\bfC_\cE)=0$ imply $H(\bfS|\bfC_\cE) = H(\bfS)$.  Indeed, $ H(\bfS| \bfC_\cE) \leq H(\bfS)$, and
\begin{subequations}\label{secrecy_lemma}
	\begin{align}
	H(\bfS) + H(\bfR)  &= H(\bfS | \bfR) + H(\bfR) \nonumber \\
	&=  H(\bfS,\bfR)  = H(\bfS,\bfC_\cE,\bfR)  \nonumber \\
	&= H(\bfC_\cE) + H(\bfS,\bfR | \bfC_\cE)  \nonumber \\
	&= H(\bfC_\cE) + H(\bfR | \bfS, \bfC_\cE)  + H(\bfS| \bfC_\cE)  \nonumber \\
	&= H(\bfC_\cE)  + H(\bfS| \bfC_\cE)\label{sec_lemma_assmp1} \\
	&\leq H(\bfR)  + H(\bfS| \bfC_\cE)\label{sec_lemma_assmp2}
\end{align}
\end{subequations}
where \cref{sec_lemma_assmp1,sec_lemma_assmp2} follow from the assumptions  $H(\bfR|\bfS,\bfC_\cE)=0$ and $H(\bfC_\cE)\leq H(\bfR)$ respectively.
On the other hand, assuming that the eavesdropper also knows $\bfs$ (in addition to $\bfc_\cE$), she/he has 
\begin{equation*}
	\tilde{c}_i = c_i - \sum_{j={1}}^{k} s_j \tilde{\alpha}_i^{q^{\ell+j-1}} = \sum_{j={1}}^{\ell} r_j \tilde{\alpha}_i^{q^{j-1}}, i\in \cE .
\end{equation*}
Since $B = [\tilde{\alpha}_i^{q^{j-1}}]_{i\in \cE, j\in [\ell]}$ is full rank, the eavesdropper can compute $[\tilde{c}_1 \ldots \tilde{c}_\ell] B^{-1} = [r_{1} \ldots r_{\ell}] $.  Thus, $H(\bfR|\bfS,\bfC_\cE)=0$. Now $H(\bfC_\cE)\leq H(\bfR)$, since $\abs{\cE} \leq \ell$. Therefore, we have an $(n,k,\ell,m,r)_{q^N}$-secret sharing scheme.
\end{IEEEproof}

\subsection{Constructions with small alphabet size: equivalence with maximal recoverability}
Note that, the size of the alphabet/shares in the construction of optimal secure scheme using Gabidulin codes is exponential in the number of nodes. 
In this section, our aim is to show that the construction of an optimal secure scheme with small alphabet size will amount to finding a {\em maximally recoverable code} 
over that alphabet. We use the construction in \cref{construction_lin_code} to form a secure scheme from an optimal LRCs with a small alphabet and analyze the conditions for that construction to satisfy \cref{lemma1}. We assume $(r+1) | n$ i.e.  $r+1$ divides $n$ for simplicity in this subsection.


We will need the following definition of maximally recoverable codes \cite{GopalanHJY13}.
{\definition \label{maximally_recov} 
Consider an $(n,k,r)_q$-optimal LRC. Let $\cQ_j : \abs{\cQ_j}=r+1, j\in [n/(r+1)]$ denote a partition of $[n]$ such that the recovery set of $i$th coordinate is,
\begin{equation}\label{part_recov}
	\cR_i = \cQ(i)\setminus \set{i}, \,\, \forall i \in [n], 
\end{equation}
where $\cQ(i) \in \set{\cQ_j}_{j}$ is the partition containing node $i$. 
Denote such an LRC by $(n,k,r,\set{\cQ_j}_{j})_q$. 
The $(n,k,r,\set{\cQ_j}_{j})_q$ LRC is called maximally recoverable  if the code obtained by puncturing any one symbol from each $\cQ_j$ is maximum distance separable (MDS).}

Note that, in  \cite{gopalan2012locality}, it was pointed out that an optimal linear LRC must have the recovery structure as in \cref{part_recov}. 

The main objective of this section is to show that the immediate construction of $(n,k,\ell,m,r)$-secret-sharing scheme from an optimal LRC is effective if and and only if  the code is maximally recoverable. 

%

{\lemma \label{lemma2} For any linear $(n,k+\ell,r,\set{\cQ_j}_j)_q$ -optimal LRC code with a generator matrix $G \in \ff_q^{n\times \br{k+\ell}}$ consider $\cS \subseteq [n] : \abs{\cS}=\ell \mbox{ and } \abs{\cS\cap \cQ_j} \leq r, j \in [n/(r+1)]$. Then, the rows corresponding to $\cS$ in $G$ are linearly independent for any $\ell$ such that
\begin{align}\label{eq:lemma2}
	\ell &\leq r-1+\br*{r\floor*{\frac{k}{r-1}}-k} 
\end{align}}
\begin{IEEEproof} 
 Partition $\cS$ as follows, $ \cS = \bigcup_{j\in [n/(r+1)]} \cS_j$ with $\cS_j = \cS\cap \cQ_j$ and let $\Lambda \defeq \set{j: \cS_j \ne 0}$. Consider a set $\cS^\prime\supset \cS : \abs{\cS^\prime}\leq k+\ell$ and define ${\cS^\prime}_j  \defeq {\cS^\prime} \cap \cQ_j$. Suppose that we can construct $\cS^\prime$ with $\cS^\prime_j \leq r, \forall j\in [n/(r+1)]$ such that the number of partitions $\cQ_j$ that contain $r$ co-ordinates of $\cS^\prime$ is at least $\ceil{(k+\ell)/r} -1$. Let $\Psi \defeq {\set{j:\cS_j^\prime = r }}$. Thus, 
\begin{equation}\label{constraints_S}
	\abs{\Psi} \geq \ceil{(k+\ell)/r} -1
\end{equation}

Construct a set $\cS^{\prime\prime} \supseteq \cS^\prime$ by adding $k+\ell-\abs{\cS^\prime}$ co-ordinates to ${\cS^\prime}$ such that, $\abs{S^{\prime\prime}\cap \cQ_j}\leq r, \forall j\in [n/(r+1)]$. Now at least $\abs{\Psi}$ more co-ordinates are recoverable from $\cS^{\prime\prime}$. Note that the input $\bfa$ for $(n,k+\ell,r,\set{\cQ_j}_j)_q$-optimal LRC is recoverable from any  $m=\br{k+\ell}+\ceil{\br{k+\ell}/r}-1$ co-ordinates and $\abs{\cS^{\prime\prime}}+\abs{\Psi} \geq m$. Thus, $\bfa$ is recoverable from $\bfc_{\cS^{\prime\prime}}$. Now, since $\abs{S^{\prime\prime}}=k+\ell$ the rows of $G$ corresponding to $\cS^{\prime\prime}$ (and hence $\cS$) must be L.I. We are now left with the task of constructing a set $\cS^\prime$ satisfying \cref{constraints_S} for the given $\cS$ with $\abs{\cS} = \ell$ satisfying \cref{eq:lemma2}. The construction is given below.

For $\abs{\Lambda} \leq k/\br{r-1}$ we can easily construct $\cS^\prime$.  Since  $\abs{\Lambda} \leq k/\br{r-1} \implies \abs{\Lambda} r \leq k+\ell$, we can choose $\Psi (\supseteq \Lambda) : \abs{\Psi} = \floor{\frac{k+\ell}{r}}$. Now to each of the partitions  $\set{\cS_j}_{j\in \Psi}$ add $r-\abs{\cS_j}$ co-ordinates from $\cQ_j$ to get a set $\cS^{\prime}$ of size $r \floor{\br{k+\ell} / r} \leq  k+\ell$. It is easy to see that this set satisfies \cref{constraints_S}.

 Now assume that $\abs{\Lambda} > k/\br{r-1}$. Choose any $\Psi \subseteq \Lambda : \abs{\Psi} = \floor{k/(r-1)}$. Select any $r-\abs{\cS_j}$ co-ordinates from $\cQ_j$ for all $j \in \Psi$. Adding these co-ordinates to $\cS$, we get $\cS^\prime$ satisfying $\abs{\cS^\prime} \leq \floor{k/(r-1)} (r-1) + \ell \leq k+\ell$. 
 Thus, from \cref{eq:lemma2} we have, 
\begin{align*}
\abs{\Psi}  +1& - \ceil{\br{k+\ell}/r} \geq \floor{k/\br{r-1}} - \frac{k+\ell}{r}  \\ 
																		 &\geq  \floor{k/\br{r-1}} - \frac{k}{r} - (1+\floor{k/\br{r-1}}-k/r -1/r) \\
																		 &= -(1-1/r) 
\end{align*}
	Since $\abs{\set{\cQ_j : \abs{\cQ_j \cap \cS^\prime} =r}} +1 - \ceil{(k+l)/r}$ is an integer, $m^\prime +1 - \ceil{(k+l)/r}\geq 0$, $\cS^\prime$ satisfies \cref{constraints_S}.
\end{IEEEproof}

For $\ell<r$, the construction (in \cref{construction_lin_code}) using an optimal LRC  code is $\ell$-secure since any $\ell$ rows of $G_1$ form an $\ell \times \ell$ Vandermonde matrix. For $\ell>r$, we have the following result, using \cref{maximally_recov} and \cref{lemma2}.

{\theorem \label{secure_tamo_barg} Consider a linear $(n,k+\ell,r,\set{\cQ_j}_j)_q$ -optimal LRC $\cC$. Then the construction in \cref{construction_lin_code} using code $\cC$ is $\ell$-secure if there exists $\cC^\prime \subseteq \cC$ of dimension $\ell$ such that $\cC^\prime$ is maximally recoverable. Conversely, if the construction in \cref{construction_lin_code} is $\ell$-secure then there must exist a maximally recoverable code $\cC^\prime \subseteq \cC$ of dimension $\ell$, for $\ell \leq r-1+\br*{r\floor{k/(r-1)}-k}$}

\begin{IEEEproof}
Let $G  = [G^1 \;G^2] \in \ff_q^{n\times (k+\ell)}$ be the generator matrix of $\cC$ where $G^1 \in \ff_q^{n \times \ell}$. Let $G^1$ be the generator matrix of a maximally recoverable code $\cC^\prime$. Consider a set $\cD \subseteq [n]$ of any $\ell$ linearly dependent rows of $G^1$. Since $\cC^\prime$ is maximally recoverable, $\cQ_j \subseteq \cD$ for at least one $j\in [n/(r+1)]$. Hence, the corresponding rows in $G$ must also be linearly dependent. Thus, from \cref{lemma1} the secret sharing construction in \cref{construction_lin_code} must be $\ell$-secure.

Now, suppose that $\cC$ does not contain any subcode of dimension $\ell$ which is maximally recoverable. Then, the code generated by $G^1$ is not maximally recoverable. Thus, there would exist an $\cS \subseteq [n]: \abs{\cS}=\ell \mbox{ and } \abs{\cS \cap \cQ_j} \leq r, \forall j\in [n/(r+1)]$ such that the rows in $G^1$ corresponding to $\cS$ are linearly dependent. Now from \cref{lemma2} we know that the rows corresponding to $\cS$ in $G$ are not linearly dependent for $\ell \leq r-1+\br*{r\floor{k/(r-1)}-k}$. Hence, from \cref{lemma1} the secret sharing scheme cannot be $\ell$ secure.	
\end{IEEEproof}


Recently an optimal construction of locally repairable codes was proposed in \cite{barg2013family} by Tamo and Barg for general values of the parameters $n,k,$ and $r$ and alphabet size of $O(n)$. Our \cref{secure_tamo_barg} implies that the secret sharing scheme constructed in \cref{construction_lin_code} using such code is $\ell$-secure if and only if the Tamo-Barg codes are maximally recoverable. In general these codes are not maximally recoverable.
It should be noted that, it is quite a nontrivial open problem to construct maximally recoverable codes with linear or even polynomial (in blocklength) alphabet size \cite{GopalanHJY13}.


In the next two sections we extend the notions and results of \cref{sec:converse_results} to
other generalized repair conditions related to distributed storage.

\section{Security for Schemes with cooperative repair} 
\label{sec:schemes_for_co_operative_repair}

Cooperative repair for a locally repairable scheme addresses simultaneous multiple failures in a distributed storage system  \cite{rawat2014cooperativelocal}\footnote{There is a related notion of cooperative recovery in regenerating codes \cite{shum2013cooperative} and security in such systems \cite{koyluoglu2014secure}. In this paper we are concerned with only the local recovery problem, and not the regenerating problem.}. To this end, we extend the definition in \cref{cond3} to a $(r,\delta)$ scheme where any $\delta$ --instead of just one-- shares can be recovered from $r$ other shares.

{\definition \label{coop_repair} A set $\cC \subseteq \ff_q^n$ is said to be $(r,\delta)$-repairable if for every $\Delta \subseteq [n] : \abs{\Delta}\leq \delta$ there exists a set $\cR(\Delta) \subseteq [n]\setminus  \Delta : \abs{\cR(\Delta)} \leq r$ such that for all $\bfc, \bfc' \in \cC$,
\begin{equation}
	{\bfc}_{\Delta} \ne {\bfc'}_{\Delta} \implies {\bfc}_{\cR(\Delta)} \ne {\bfc'}_{\cR(\Delta)}
\end{equation}
}

Using \cref{coop_repair} we can generalize the notion of an $(n,k,\ell,m,r)_q$-secret sharing scheme. 
For this system we derive an upper bound on the capacity $k$ given $n, m, \ell,r, \mbox{ and }\delta$.

{\definition An $(n,k,\ell,m ,\br{r,\delta})_q$-secret sharing scheme consists of a randomized encoder $f(.)$ that stores a file $\bfs \in \ff_q^k$ in $n$ separate shares, such that the scheme is $(r,\delta)$-repairable (\cref{coop_repair}), satisfies the recovery condition (cf. \cref{cond1}) and $\ell$-secure (cf. \cref{cond2}).}

\subsection{The case of m =n}
Error-correcting codes with $(r,\delta)$-repairability were considered in \cite{rawat2014cooperativelocal} ($\ell=0$ or no security) and the following upper-bound on the rate of such codes has been proposed,
for the case of $m =n$.
\begin{equation}\label{bounds_prev_coop}
	R = \frac{k}{n} \leq \frac{r}{r+\delta}.
\end{equation}
For the case of $\ell$-secure codes we give an analogous upper bound on  the rate of a secret sharing scheme in the following.

{\theorem \label{coop_theorem} The rate $R=k/n$ of an $(n,k,\ell,n, (r,\delta))_q$ secret sharing scheme is bounded as,
\begin{equation}\label{coop_theorem_UB}
	R \leq \frac{r}{r+\delta} - \frac{\ell}{n}.
\end{equation}
}
\begin{IEEEproof}
	For an $(n,k,\ell,(r,\delta))_q$ scheme we construct a set of size $m=n$ similar to \cref{algo:conver_result} except instead of choosing a set of size $1$ in steps 2 and 5, we find a set of size $\delta$. Then using the same arguments we must have at least $\nu = m/(r+\delta)$ number of steps. Hence, subtracting the number recoverable symbols $\delta \nu$ from the $m$ symbols we must have,
	\begin{align*}
		k+\ell &\leq m -  \delta \nu = n -  \delta \frac{n}{r+\delta} \\
		\implies \frac{k+\ell}{n} &\leq \frac{r}{r+\delta}.
	\end{align*}
\end{IEEEproof}

\noindent{\em Construction:} Note that,  any linear $q$-ary $(r,\delta)$-repairable error-correcting code of length $n$ and dimension $k$  will give rise to a $(n,k,0, (r,\delta))$-secret sharing scheme. In \cite[Sec.~6]{rawat2014cooperativelocal}, an $(r,\delta)$ repairable code has been constructed using bipartite graphs of large girth.
In particular, that construction results in parameters such that
$$
\frac{k}{n} \ge \frac{r-\delta}{r+\delta}.
$$
It can also be seen from the discussion of \cref{sec:constr} that Gabidulin precoding (\cref{Gabidulin_precoding}) would give an  $\ell$-secure construction with alphabet $\ff_{q^N}$,  $N\geq n$, from any optimal linear $(n,k+\ell,0,(r,\delta))_q$-secret sharing scheme. Thus, for any $(n,k+\ell,0,(r,\delta))_q$ secret sharing scheme achieving the upper-bound in \cref{bounds_prev_coop} we can achieve the corresponding upper-bound in \cref{coop_theorem}.  Hence, using the code  of \cite[Sec.~6]{rawat2014cooperativelocal} in conjunction with the Gabidulin precoding,  it is  possible to obtain a rate of
$$
\frac{k}{n} \ge \frac{r-\delta}{r+\delta} - \frac{\ell}{n},
$$
which is an additive term of $\frac{\delta}{r+\delta}$ away from the optimum possible.

\subsection{The case of $m < n$.}
The bound for general case of $m <n$ can be deduced from the same arguments as above. In fact,
by slightly generalizing \cref{algo:conver_result}, we get the following result: for any $(n,k,\ell,m ,\br{r,\delta})_q$-secret sharing scheme ,
\begin{equation}\label{coop_conver_result}
  k+\ell \leq m-\floor*{\frac{m}{r+\delta}}\delta - h
\end{equation}
where $h =\br*{m \mod (r+\delta)-r}^+$ and $x^+ \defeq \begin{cases} 
      0 & x\leq 0, \\
      x & x> 0. \\
   \end{cases}$
   
Note that, this results in slightly weaker bound for the case of $m=n$ than \cref{coop_theorem_UB}.
In general for $m<n$ and arbitrary values of $\ell$, we do not have any good construction that will be close to the bound. While the expander-graph based 
constructions of $(r,\delta)$-locally repairable codes from \cite{rawat2014cooperativelocal} can be generalized, their performance is very far from the bound of
 \cref{coop_conver_result}.

\section{Security for repairable codes on graphs} 
\label{sec:security_for_repairable_codes_on_graphs}
Another extension of local repair property for distributed storage has recently been proposed in  \cite{mazumdar2013duality, mazumdarachievable}. 
Consider a Distributed Storage System as a directed graph $\cG$ such that a node of the graph represents a node of the Distributed Storage System and each node can connect to only its out-neighbors for repair. We define an $\ell$-secure code in this scenario as follows.

\subsection{Repairable Codes on Graph}

{\definition \label{graph_code} Let $\cG = ([n],E)$ be a  graph on $n$ nodes. An $(n,k,\ell,m,\cG)_q$-secret sharing scheme consists of a randomized encoder $f$ that can store a uniformly random secret $\bfS \in \ff_q^k$ on $n$ shares/nodes, $\bfC=f(\bfS), \bfC \in \ff_q^n$, such that the system is $\ell$-secure (cf. \cref{cond2}) and the data can be recovered from any $m$ shares (cf. \cref{cond1}). In addition the share of  any node can be recovered from its neighbors i.e. 
$$H(\bfC_i|\bfC_{N(i)})=0$$
where $N(i) = \{j \in [n]: (i,j) \in E\}$ denotes the neighbors (out-neighbors in the case of a directed graph) of node $i$ in the graph $\cG = ([n],E)$.}

A  bound on the capacity of such a scheme in directed graphs for $\ell=0$ (no security) was derived in \cite{mazumdar2014storage},
\begin{equation}\label{mazumdar_rdss}
	m  \geq k +\!\! \max_{\substack{U \in \cI(\cG) :\\ \abs{N(U)} \leq k-1}} \abs{U} 
\end{equation}
where $\cI(\cG)$ denotes the set of induced acyclic subgraphs in $\cG$, and $N(U) \defeq \cup_{i\in U} N(i)\setminus U$ denotes the neighbors of $U$. 
For undirected graphs we have the same bound with $\cI(\cG)$ denoting the collection of all independent sets of the graph.
The lower bound on $m$ for an $\ell$-secure scheme on a graph $\cG$ is given in the following.

{\theorem For any $(n,k,\ell,m,\cG)_q$-secret sharing scheme on a directed graph $\cG$, $m$ satisfies the following lower bound,
\begin{equation}\label{RDSS_ub}
	m \geq k+\ell +\!\! \max_{\substack{U \in \cI(\cG) :\\ \abs{N(U)} \leq \ell+k-1}} \abs{U} 
\end{equation}
where $\cI(G)$ denotes the set of induced acyclic graphs in $\cG$.}
\begin{proof}
	Since any $m$ co-ordinates in the shares $\bfC = (C_i)_{i\in[n]}$ can recover the secret $\bfS$ we must have,
	\begin{equation}\label{intermediate_triv}
		m\geq \abs{W}+1
	\end{equation}
	for all $W \subseteq [n]$ such that the $H(\bfS | \bfC_{W}) > 0$. Let $U$  be an acyclic subgraph $U \in \cI\br{\cG}$, such that $N(U)\leq \ell+k-1$. Construct a set $V \supseteq \set{U \cup N(U)}$ by adding any $\ell+k-1 - \abs{N(U)}$ nodes to $U \cup N(U)$.	Thus, $\abs{V} = k+\ell+\abs{U}-1$. We show that $H(\bfS | \bfC_{V}) > 0$ for any such $V$.

	Note that for any three random $X,Y,Z$ variables we must have,
	\begin{align}\label{general_infotheo}
		H(X|Y,Z)& = H(X,Z|Y) - H(Z|Y)\nonumber\\
		 &= H(X|Y)+H(Z|X,Y) - H(Z|Y) \nonumber\\
		 &\geq H(X|Y) - H(Z).
	\end{align}

	Assume that the eavesdropper selects an $\ell$-subset $\cE \subseteq [n]$ in the set $V$. Then, since the eavesdropper must not get any information about the secret,
	\begin{equation}\label{eaves_dropper}
		H(\bfS | \bfC_\cE) = H(\bfS)
	\end{equation}
	Since the sub-graph $U$ is acyclic the nodes in $U$ must be a function of the leaf nodes and the nodes in $N(U)$. Now, the leaf nodes must also be a function of $N(U)$ since their out-neighbors can only be in $N(U)$. Therefore,
	\begin{align*}
		H(\bfS|\bfC_V)& = H(\bfS|\bfC_{N(U)}) = H(\bfS | \bfC_{\cE}, \bfC_{N(U) \setminus\cE}) \\
		&\overset{(a)}{\geq} H(\bfS | \bfC_{\cE}) - H(\bfC_{N(U)\setminus\cE}) \\
																			 &\overset{(b)}{=} H(\bfS) - H(\bfC_{N(U)\setminus\cE}) \\
																			 &\overset{(c)}{>}0
	\end{align*}
	where $(a)$ and $(b)$ follow from \cref{general_infotheo} and \cref{eaves_dropper} respectively, and $(c)$ is is true since $\abs{N(U)\setminus\cE}= k-1$.
\end{proof}


When  $m=n$,  i.e. when the scheme does not need to protect against catastrophic failures, we can formulate a converse bound for repairable codes on graphs that does not follow directly from the above theorem. 

{\theorem Consider an $(n,k,\ell,n,\cG)_q$ secret sharing scheme. The secrecy capacity of the scheme satisfies the following upper-bound.
\begin{equation}\label{secrecy_no_distance}
	k \leq n - \abs{U} - \abs{\ell}
\end{equation}
where $U$ is the largest acyclic induced subgraph in $\cG$ when $\cG$ is a directed graph, and it is the largest independent set when $\cG$ is undirected.}
\begin{proof}
We will show the proof for $\cG$ directed.
	Consider the shares $\bfC_{U}$ corresponding to the nodes in $U \subseteq [n]$.  The recovery set of any node in $U$ can contain its children in $U$ or co-ordinates in $[n]\setminus U$. Since $U$ is ayclic, all the leaf nodes of $U$ have recovery sets in $[n]\setminus U$. Thus, we can recover all the leaf nodes from the co-ordinates in $[n]\setminus U$. Now, we can recursively recover all the co-ordinates of $U$ from the co-ordinates in $[n]\setminus U$. Thus,
	\begin{equation}\label{t0}
		H(\bfC_U | \bfC_{[n]\setminus U}) = 0
	\end{equation}

	\Cref{t0} is true because all the leaf nodes in $U$ must have their recovery sets in $[n]\setminus U$. And by recovering the leaf nodes we can recover all nodes in $U$. 	Now, since $H(\bfS | \bfC)= 0$ we must have from \cref{t0},
	\begin{equation}\label{t1}
		H(\bfS|\bfC_{[n]\setminus U}) = 0 
	\end{equation}
	Now, suppose that the eavesdropper selects an $\ell$-subset $\cE \in [n]\setminus U$. Then, we must have,
	\begin{equation}\label{t2}
		H(\bfS) = H(\bfS | \bfC_{\cE})
	\end{equation}
Therefore, using \cref{t1,t2} we have,
\begin{align*}
	H(\bfC_{[n]\setminus U}| \bfC_{\cE} )  &= H(\bfC_{[n]\setminus U}| \bfC_{\cE} )  + H(\bfS| \bfC_{[n]\setminus U}, \bfC_{\cE} )	\\
	&=  H(\bfS, \bfC_{[n]\setminus U}| \bfC_{\cE} ) \\
	&= H(\bfS| \bfC_{\cE} ) + H(\bfC_{[n]\setminus U} | \bfS, \bfC_{\cE} ) \\
	&= H(\bfS) + H(\bfC_{[n]\setminus U} | \bfS, \bfC_{\cE} ) \\
	\implies H(\bfS)  &= H(\bfC_{[n]\setminus U}| \bfC_{\cE} ) - H(\bfC_{[n]\setminus U} | \bfS, \bfC_{\cE} ) \\
	\implies H(\bfS)  &\leq H(\bfC_{[n]\setminus U} | \bfC_{\cE}) \leq n - \abs{U} - \ell .
\end{align*}
\end{proof}

Note that the bound in \cref{secrecy_no_distance} parallels the feedback vertex set upper-bound in \cite[Prop.~11]{mazumdar2014storage}. Here, a feedback vertex set of a graph is a set of nodes such that every cycle in the graph has a vertex in the set.

\subsection{Achievable Schemes for Secure Repairable Codes on Graphs} 
\label{sec:gabidulin_codes}
In this section we consider construction of  $(n,k,\ell,m,\cG)_q$-secret sharing scheme only when $m=n$. We do not have any nontrivial construction
for the case of $m <n$.

Consider a secret sharing scheme for the case of undirected graphs 
(\cref{graph_code}). A maximum matching $\cM(\cG)$ of the graph $\cG$ is defined as the set of edges of maximum cardinality such that no two edges have a vertex in common. To construct a recoverable scheme for this code, with input  $\bfx \in \ff^{\abs{\cM(\cG)}}$, we assign a coordinate of $\bfx$ to both vertices for every edge in $\cM(\cG)$. For recoverability, we note that a symbol in  vertex $v$ can be recovered from $u$, where $(v,u) \in \cM(\cG)$. 

Suppose $\abs{\cM(\cG)} = k+\ell$. Consider the vector input $\bfx \in \ff^{k+\ell}$ to the above scheme. We set $\bfx = G \times [\bfs\;\; \bfr], \bfs \in \ff^k, \bfr \in \ff^\ell$, where $\bfs$ is the secret, $\bfr$ is an instance of a uniform random vector, and $G$ is the $(k+\ell)\times(k+\ell)$ Vandermonde matrix $G=[\alpha_i^{j-1}]_{ij}$ with $\set{\alpha_i}_i$ distinct elements in $\ff_q$. Thus, from \cref{lemma1}, we see that this scheme is $\ell$-secure as well as
recoverable.

The capacity of this scheme is $k =  \abs{\cM(\cG)} -\ell \ge \frac{n -\abs{U}}{2} -\ell,$ where $U$ is the maximum independent set. This is true since 
if we remove both end-vertices of the edges of the matching then we are left with an independent set. Compared to \cref{secrecy_no_distance}, we are
an additive term of at most $\frac{n -\abs{U}}{2}$ away from what is the maximum possible.

For directed graphs $\cG = ([n],E)$ we use the repairable codes presented in \cite{mazumdar2014storage} below to construct a secure scheme. 
Suppose that the graph has $K :=k+\ell$ vertex disjoint cycles. 
Then it is easy to see that we can form a locally repairable scheme capable of  storing  $k+\ell$  symbols (one symbol per cycle) by repeating the same symbol on every vertex in a cycle. Hence, it is possible to store as many symbols as the maximum number of vertex disjoint cycles in the graph. In \cite{mazumdar2014storage}, it was  shown that we can do better by using vector codes. We describe below the vector linear LRC codes constructed in \cite{mazumdar2014storage}. 

Consider the set $\cP$ of all cycles in $\cG([n],E)$. Suppose, $\Pi : \cP \rightarrow \rationals$ assigns a rational number to every directed cycle. Let $V(C), C\in \cP$ denote the vertices of the cycle $C$. Let $K$ denote the maximum value of $\sum_{C\in \cP} \Pi\br{C}$, over all such mappings $\Pi$, under the following constraint,
\begin{equation*}
	\sum_{C:i \in V(C)} \Pi(C) \leq 1, \;\;\forall i \in [n] . 
\end{equation*}
Let the optimal assignment $\Pi$ on $\cP$ be denoted as $\Pi(C) = \frac{n(C)}{p}$, where $n(C), p \in \integers^+$. It is possible to find this optimum 
by solving a linear program.
Then \cite{mazumdar2014storage} constructs a vector LRC for the graph $\cG$ in $\ff_q$ with storage capability of  $pK$ symbols
 and per node storage equal to $p$ symbols.

Let $\bfs \in \ff_{q}^{pk}, \bfr \in \ff_q^{p\ell}$ represent the secret and an instance of a uniform random vector, respectively. We obtain  $\bfx \in \ff_q^{pK}, K :=k+\ell,$ by  $\bfx = G\times [\bfs\;\; \bfr]$, where $G$ is a $pK\times pK$ Vandermonde matrix $G=[\alpha_i^{j-1}]_{ij}$ with $\set{\alpha_i}_i$ distinct elements in $\ff_q$. 
$\bfx$ is then stored in the graph using the scheme described above.
Since an $\ell$-strength eavesdropper can only observe at most $p \ell$ co-ordinates in $\bfa$, we can use \cref{lemma1} to see that the scheme is $\ell$-secure as well as recoverable. 

It is known (cf. \cite{mazumdar2014storage}) that,
$4K  \ln 4K \ln \log_2 4K \geq  n-|U|,$ for $U$ being the maximum acyclic induced subgraph. 
Hence, we must have, 
$$k  \geq  \frac{n-|U|}{ c\log n \log \log n}-\ell.$$
However this achievability result  is quite far away from the bound of \cref{secrecy_no_distance}.


\section{Perfect Secret Sharing and General Access Structures}\label{sec:largest_share_size}
So far in this paper we were concentrating on a secret sharing scheme that is not perfect, i.e., the access structure and
the block-list are not complementary. 
In this section we provide results regarding existence of locally repairable of perfect secret sharing schemes 
and the relation between sizes of shares and secret in those schemes.

\subsection{Perfect access structures with locality}\label{perfect_as_with_loc}
To make the $(n,k,\ell,m,r)$ secret sharing scheme perfect, we must have $m = \ell+1$. This results in a threshold secret-sharing scheme. Now, from \cref{easy_bound} we have,
\begin{equation*}
	k \leq 1-\floor*{\frac{\ell+1}{r+1}}.
\end{equation*}
Thus, for storing any secret we must have $r \geq \ell+1 = m $. Since any secret sharing scheme works when $r \geq m$ (local repair in this case imply full revelation of secret) only trivial locally repairable codes are possible for  threshold secret sharing schemes. This implies the following statement.
{\proposition A threshold secret sharing scheme is not locally repairable.} 


 Note that,  perfect secret sharing schemes are a natural generalization of threshold schemes. Although for threshold schemes the locality cannot be small/nontrivial,  we show that this is not true for general access structures and perfect schemes.
Indeed, the following is true.

{\proposition There exists an access structure $\cA_s$, for which   a perfect secret sharing scheme is possible with arbitrary non-trivial locality $r$ i.e. $r < \min_{A\in \cA_s} \abs{A}$.}
\begin{IEEEproof}
Let $n,\kappa$ be such that $r|\kappa$ and $(r+1)|n$. Consider an $(n,\kappa,r,\set{\cQ_j}_j)$  maximally recoverable LRC (\cref{maximally_recov}). We know that such codes exist from \cite{GopalanHJY13}. Now, we use the Gabidulin precoding method described above to construct a $(n,k=1,\ell=\kappa-1,m=\kappa(1+1/r),r)$ secret sharing scheme from this code. 

Define the access structure to be $\cA_s = \set{A \subseteq [n] \;: \;\; \sum_{j=1}^{n/(r+1)} \min\set{\abs{A\cap \cQ_j}, r} \geq \kappa}$. Now given any $A\in \cA_s$, a user accessing the shares corresponding to $A$ can determine the secret $s_0$ because the set always contains $k$ shares of a punctured $(nr/(r+1),\kappa)$-MDS code.

For a perfect secret sharing scheme the block-list is given by $\cB_s = \set{B \;: \;\; \sum_{j=1}^{n/(r+1)} \min\set{\abs{B\cap \cQ_j}, r} < \kappa}$. Assume that the eavesdropper has access to a set $B \in \cB_s$. Construct the following set of size at most $\kappa-1$ from $B$,
$$B^\prime = \cup_{j=1}^{n/(r+1)} N_j^\prime, B^\prime \ \subseteq B$$
where $N_j^\prime \subseteq N_j, N_j = B\cap \cQ_j$ is obtained by removing any one co-ordinate if $\abs{N_j}>r$, otherwise $N_j^\prime = N_j$. Note that $\abs{B^\prime} < \kappa$. Since all the shares in $B$ are recoverable from $B^\prime \subseteq B$, an eavesdropper with access to the nodes in $B$ is equivalent to an eavesdropper with access to $B^\prime$. And since $\abs{B^\prime} \leq \ell = \kappa-1$, the eavesdropper does not get any information about the secret.
\end{IEEEproof}

Can the above proposition be made general? Is it possible to characterize the locality for general secret sharing schemes? 
Shamir's \cite{shamir1979share} perfect threshold secret sharing scheme for the access structure $\cA_s = \set{A \subseteq [n] : |A|\geq k}$
is one of the first general construction of secret sharing protocols.
The scheme is defined for a scalar secret $s \in \ff$ and a set of $n$ participating nodes $P$. The scheme uses an $(n,k)$ Reed Solomon code defined using the polynomial $\sigma(x) = s + \sum_{i=1}^{k-1}r_i x^i$, where $r_i$ are instances of uniform random variables in $\ff$.

Ito, Shaito, and Nishizeki \cite{Ito1987} define a generalization of Shamir's scheme that works for arbitrary monotone access structures. \label{maximal_minimal_sets} Define a maximal element $B \in \cB$ as a set such that $A\supsetneq B \implies A \notin \cA$. Similarily,  define a minimal set $A \in \cA$ as a set such that $B\subsetneq A \implies B\notin \cA$. Consider the set of maximal elements of the block-list $\cB$, denoted $\cB^{\dagger}$. The scheme uses the generator polynomial $ \sigma(x) = s + \sum\limits_{i=1}^{|\cB^\dagger|-1}r_i x^i$ to generate  $\abs{\cB^{\dagger}}$ shares $\set{c_B}_{B \in \cB^{\dagger}}$ -- one share corresponding to each maximal set in $\cB$. The shares are distributed such that each user gets the shares corresponding to the subset it does not belong to, i.e. participant node $p$ gets the shares 
\begin{equation}\label{node_shares}
	\set{c_B : p \notin B, B \in \cB^{\dagger}}
\end{equation}

Now, suppose that share of a node $p$ is lost in a secure code with participants $P$ and block-list $\cB$. To recover the share of $p$ we access the shares of participants in the set $\cR(p)$ where the optimal set $\cR(p)$ is
\begin{equation}
	\cR(p) = \min_{{R:\forall B \in \cB^{\dagger}, p \notin B \; R \not\subseteq B}} |R|.
\end{equation}
To have non-trivial locality, one must have $\max_p |\cR(p)|$ to be strictly less than the maximal sets in the block-list.  


\subsection{Size of a share for perfect secret sharing with locality}\label{sec:largest_share_sizech} 

We know that, for perfect secret sharing schemes, the size of the secret cannot be larger than the size of a share \cite[Lemma 2]{Beimel_secret_sharing_schemes}. Let us see why this statement is true. 
	Let the secret $\bfs$ belong to a domain $\cK$ and the share of node $j$ belong to $\cK_j$. Assume that there exists a perfect secret sharing scheme which realizes the access structure $\cA$ when $\abs{\cK}<\abs{\cK_j}$. Let $B\subseteq [n]$ be a minimal set in $\cA$ such that $j \in B$. Define $B^\prime = B\setminus \set{j}$. Then, since the secret sharing scheme is perfect, for every value of the the shares in $B_j$ all secrets in $K$ must have the same probability. Thus, since the value of the shares of $B$ determine the secret completely there must exist an injective mapping from $K$ to $K_j$. But since $\abs{K_j} < \abs{K}$ this cannot be possible.

In \cite{Csirmaz1997}  the minimum node storage required for arbitrary monotone access structures is analyzed. In that paper,   an access structure was constructed for which the sizes of the shares has to be $n/log(n)$ times the size of the secret for any perfect scheme. For secret sharing schemes with local repairability and fixed recovery sets, all monotone access structures are not feasible. The minimal sets of the access structure cannot include any recovery set. Here, we extend the result in \cite{Csirmaz1997} to the restricted class of monotone access structures. 

Assume $(r+1)|n$. Suppose that the secret denoted by the random variable $S$ is stored on $n$ shares as  $C_i, i\in[n]$ and the shares have locality $r$ (\cref{cond3}). Consider a partition of $[n]$, $\cQ_j : {\cQ_j}, j \in [n/(r+1)]$ such that the recovery sets are given by  \cref{part_recov}.
For a perfect secret sharing scheme on $[n]$ with monotone access structure $\cA_s$, the minimal sets  $\cA_s^\star$ of $\cA_s$, must satisfy,
\begin{equation}\label{restricted_access_struct}
	A \in \cA_s^\star \implies A \not\supseteq \cQ_j \,\, .
\end{equation}
Denote this class of monotone access structures with $\bbM_s$. We have the following result for the minimum size of a share for secret sharing schemes with access structure $\cA_s \in \bbM_s$.

\begin{theorem}\label{thm:size}
Consider distribution of shares of secret $S$ to $n$ nodes with locality $r$, recovery sets as in \cref{part_recov}. Then, there is an access structure $\cA_s \in \bbM_s$ (\cref{restricted_access_struct}), such that any perfect scheme for $\cA_s$, if exists, must satisfy,
\begin{equation}
	\alpha \geq \frac{(r+1)n}{r\log n} H(S).
\end{equation}
where $\alpha$ is the average entropy of the shares.
\end{theorem} 
\begin{IEEEproof}
First, let us define a {\em polymatroid} $(Q=\set{[n],S},\phi)$ as follows,
\begin{subequations}\label{polymatroid}
	\begin{gather}
	\phi(A) = \frac{H(\bfc_A)}{H(S)}, \;A \subseteq [n] \\
	\phi(A,S) = \frac{H(\bfc_A,S)}{H(S)}, \;A \subseteq [n]
\end{gather}
\end{subequations}
A polymatroid function must satisfy the following properties,
\begin{enumerate}[label=\textbf{P\arabic*},ref=(P\arabic*)]
	\item \label[property]{prop_matroid1} $\phi(A)\geq 0$ for all $A\subseteq Q$, $\phi(\emptyset)=0$
	\item \label[property]{prop_matroid2} $\phi$ is monotone i.e. $A\subseteq B \subseteq Q$, then $\phi(A)\leq \phi(B)$
	\item \label[property]{prop_matroid3} $\phi$ is submodular i.e. $\phi(A) + \phi(B) \geq \phi(A\cup B) + \phi(A\cap B)$ for any $A,B \subseteq Q$
\end{enumerate}
Note that, the definition in \cref{polymatroid} satisfies all the conditions above. In addition, the definition satisfies the following properties,
\begin{enumerate}[label=\textbf{P\alph*},ref=(P\alph*)]
	\item \label[property]{prop_matroid1a} $\phi({A,S}) = \phi(A)$, for every $A\in \cA_s$
	\item \label[property]{prop_matroid1b} $\phi({A,S}) = \phi(A)+1$, for every $A \notin \cA_s$
\end{enumerate}
which easily follow from the recovery and the security properties i.e. $H(S|\bfc_B)=H(S)$ and $H(S|\bfc_A)=0$, $A\in \cA_s$ and $B\in \cB_s = 2^{[n]} - \cA_s$ and the definition in \cref{polymatroid}.

Using \cref{prop_matroid1,prop_matroid2,prop_matroid3} and \cref{prop_matroid1a,prop_matroid1b} we have the following result, for any $A,B \in \cA_s$ such that $A\cap B \notin \cA_s$,
\begin{gather}
	\phi(A,S) + \phi(B,S) \geq \phi((A\cup B),S) + \phi((A\cap B),S)\nonumber \\
	\implies \phi(A) + \phi(B) \geq \phi(A\cup B) + \phi(A\cap B) +1 \label{matroid_property}
\end{gather}

Consider the set $M$ of size $\eta$ such that $(r+1) | \eta$ and it contains $\eta/(r+1)$ partitions $\cQ_j$. Another set $N \subseteq [n]\setminus M: \abs{N}=\nu\defeq2^\eta-(r+2)^{{\eta /(r+1)}}+1$ is chosen such that $\abs{N\cap \cQ_j} \leq r,\;\; \forall j$. The parameter $\eta$ for the size of the sets $M,N$ is chosen to be the largest possible, i.e. the maximum $\eta$ satisfying,
\begin{equation}\label{max_m}
 \eta-\floor*{\frac{\eta}{r+1}} + 2^\eta - (r+2)^{\eta/(r+1)} +1 \leq n\frac{r}{r+1}
\end{equation}

Now, construct a sequence $\set{M_i}_{i=0}^{\nu-1}$, for $M_i  \in 2^M$ of length $\nu$, such that it satisfies the following conditions for all sets $M_i$ in the sequence,
\begin{enumerate}[label=\textbf{C\arabic*},ref=C\arabic*]
	\item \label[condition]{cond_redundancy1} If for any partition $\cQ_j, \cQ_j\cap (M_i-M_{i+1} )\ne \emptyset$ and $\abs{\cQ_j \cap M_i} \geq r$, then $\abs{\cQ_j \cap M_{i+1}} < r$
	\item \label[condition]{cond_redundancy2} $M_i \not\subseteq M_{i^\prime}, i < i^\prime$ 
\end{enumerate}

To construct the sequence $\set{M_i}_i$ of length $\nu$ satisfying \cref{cond_redundancy1,cond_redundancy2}, we first construct a sequence $\set{M^\prime_i}_{i=0}^{2^\eta-1}$, $M^\prime_i \subseteq M : \abs{M^\prime_i}\leq \abs{M^\prime_{i+1}}$. It is easy to see that all subsequences of $\set{A^\prime_i}$ satisfy \cref{cond_redundancy2}. From this sequence we remove all sets $M_i^\prime, i\geq 1$ such that $\abs{(M_0-M_i^\prime) \cap \cQ_j} \leq 1$. Note that, the number of the sets removed is,
\begin{equation*}
	\sum_{1\leq i\leq \eta/r+1} {\eta/(r+1) \choose i}(r+1)^i = (r+2)^{\eta/(r+1)}-1 .
\end{equation*}
The sequence $\set{M_i}_i$ thus constructed has length $\nu$. To see that this sequence satisfies \cref{cond_redundancy1} note that $\abs{(M_0-M_i) \cap \cQ_j} > 1, \forall i\geq1$ implies that $\set{M_i}_i$ satisfies \cref{cond_redundancy1}. Thus the constructed sequence satisfies \cref{cond_redundancy1,cond_redundancy2}.

Let $N = \set{b_1, \ldots, b_{\nu-1}}$. Define another sequence of sets $N_i = \set{b_1, \ldots, b_i}, i \in [\nu-1]$ and $N_0=\emptyset$. Consider a monotone access structure $\cA_s$ that contains the sets $U_i \defeq M_i \cup N_i, i\in \set{0, \ldots, \nu-2}$. Let the minimal sets in this access structure be,
\begin{equation}\label{minimal_access_struct}
	\cA_s^\star = \set*{A \subseteq U_i : \abs{A\cap \cQ_j} = \min\set{\abs{A\cap \cQ_j}, r}, \;\forall i \in \brac*{\frac{n}{r+1}}} .
\end{equation}
Thus, $\cA_s \in \bbM_s$.

Consider the following sets $P = N_i\cup M$ and $Q=M_{i+1}\cup N_{i+1}$. Since $P\supseteq U_i$ and $Q\supseteq U_{i+1}$, $P,Q \in \cA_s$. Now, $P\cap Q = N_i\cup M_{i+1}$. From \cref{cond_redundancy1,minimal_access_struct}, we see that there exists a set $A^\star \in \cA_s^\star, A^\star \subseteq U_i$ such that $P\cap Q \subsetneq A^\star$. Therefore, $P\cap Q \notin \cA_s$. Applying \cref{matroid_property} on $P,Q$, we have,
\begin{align}\label{diff_eq}
	&\brac*{\phi(N_i\cup M) - \phi(N_i\cup M_{i+1})}\nonumber\\
	 & \qquad - \brac*{\phi(N_{i+1}\cup M) - \phi(N_{i+1}\cup M_{i+1})} \geq 1 .
\end{align}
Using \cref{prop_matroid3} we have,
\begin{equation}\label{tmp_diff_eqn}
	\phi(N_{i+1}\cup M_{i+1}) -  \phi(N_i\cup M_{i+1}) \geq \phi(N_{i+1}) -  \phi(N_i) .
\end{equation}
Thus, combining \cref{diff_eq,tmp_diff_eqn} we have,
\begin{gather}\label{diff_eq}
	\brac*{\phi(N_i\cup M) - \phi(N_i)} - \brac*{\phi(N_{i+1}\cup M) - \phi(N_{i+1})} \geq 1 .
\end{gather}
Adding \cref{diff_eq} for $ i \in \set{0,\ldots,\nu-3}$ we have,
\begin{equation}\label{final_eq}
	\phi(M) - [\phi(N_{\nu-2} \cup M) - \phi(N_{\nu-2})] \geq \nu-2 .
\end{equation}
Thus, from the recoverability property we have $\phi(M) \leq \eta r/(r+1)\alpha$. Since, $M\in \cA_s$ and $N_{\nu-2} \not\in \cA_s$, $\phi(N_{\nu-2} \cup M) - \phi(N_{\nu-2}) \geq 1$. Thus, we have from \cref{final_eq},
\begin{equation}\label{upperbound}
	\alpha \geq (r+1)\frac{2^\eta-(r+2)^{\eta/(r+1)}}{\eta r} H(S) .
\end{equation}
Since, $\eta=\Omega(\log{n})$ and $(r+2)^{1/(r+1)} <2$ from \cref{max_m}, \cref{upperbound} asympototically (with $n$) gives,
\begin{equation*}
	\alpha \geq \br*{\frac{r+1}{r}} \frac{n}{\log{n}} H(S).
\end{equation*}
\end{IEEEproof}


\begin{appendices}
\crefalias{section}{appsec}

\section{Proof of \cref{lemma1}}\label[Appendix]{proof_lemma1}
Consider the submatrix $H_{\ell\times \br{k+\ell}}$ of $G$ corresponding to $\ell$ rows, $I_\ell \subseteq [n]$. Assume that the eavesdropper observes $I_\ell$. Wlog assume that $\rank\br{H} = \ell$, since the eavesdropper effectively observes $\rank\br{H}$ shares.

"$\impliedby$" Assume that any $\ell$ rows of $G^1$ corresponding to $\ell$ L.I. rows of $G$ are L.I. Thus, $\rank\br{H_1}=\ell$ by assumption. Let $\bfc = G \bfa$ and  $H=[H_1 \;\;H_2]$ where $H_1$ is $\ell \times \ell$ and $H_2$ is $\ell \times k$. Then,
\begin{equation}\label{full_rank_sec}
	 H_1 \bfr  = \bfc_{I_\ell} - H_2 \bfs
\end{equation}
Now, given $c_{I_\ell}$, for every $\bfs$ there is a unique solution to $\bfr = {H_1}^{-1}(\bfc_{I_\ell} - H_2 \bfs)$. Since, each of those vectors are equally probable the eavesdropper does not get any information about $\bfs$.

"$\implies$"
Conversely, suppose that $H_1$ is not full rank. (but $\rank\br{H}=\ell$ by assumption). If for a given $\bfc_{I_\ell}$ there does not exist a solution to \cref{full_rank_sec} for some $\bfs \in \ff_q^k$ then $H(\bfs | \bfc_{I_\ell}) < H(\bfs)$. This happens iff for some $\bfa \in \ff_q^{k+\ell}$,
\begin{equation}\label{fullrank_cond}
	H \bfa - \colspan(H_2) \not\subseteq \colspan(H_1)
\end{equation}
where $\colspan(.)$ denotes the column span of a matrix and $H \bfa - \colspan(H_2) = \set{H \bfa - \bfv: \bfv \in \colspan(H_2)}$. Now, $\colspan(H_2) \not\subseteq \colspan(H_1)$ since $\dim\br{\colspan({H_1,H_2})} = \ell$ and $\dim\br{\colspan(H_1)} < \ell$ by assumption. Thus, \cref{fullrank_cond} is satisfied for $\bfa=\mathbf{0}$ which implies that in this case the eavesdropper does get some information about $\bfs$.

\section{Achievability using Linear Network Codes}\label{sec:achievability_nklmr}

In this appendix, we show that the limit derived in \cref{easy_upper_bound} is achievable using a random {\em linear network code} (LNC).
The rest of this section is devoted to the proof of  \cref{thm:achievability} via
 the technique provided in  \cite{papailiopoulos2012locally}. 
We assume that $k_0$ is such that,
\begin{align}\label{init_k}
m = k_0 + k_0/r-1
\end{align}
For simplicity, further assume that $r$ divides $k_0$ and $(r+1)$ divides $n$. 

Our roadmap for the proof is the following. We 
 analyze the network flow graph in \cref{fig:flow_graph}, that has
been adapted and modified from \cite{papailiopoulos2012locally}.
We first show that this graph has multicast capacity $k_0$.
Further there exists an LNC for this graph which corresponds to an 
$(n,k_0, 0,m,r)$-secret sharing scheme. 
Then,  we  impose additional constraints on the LNC for
 the graph in \cref{fig:flow_graph} to get an $\ell$-secure scheme, i.e.,  an $(n, k=k_0-\ell, \ell,m,r)$-scheme.
Clearly this satisfies \cref{eq:trivial}.

%

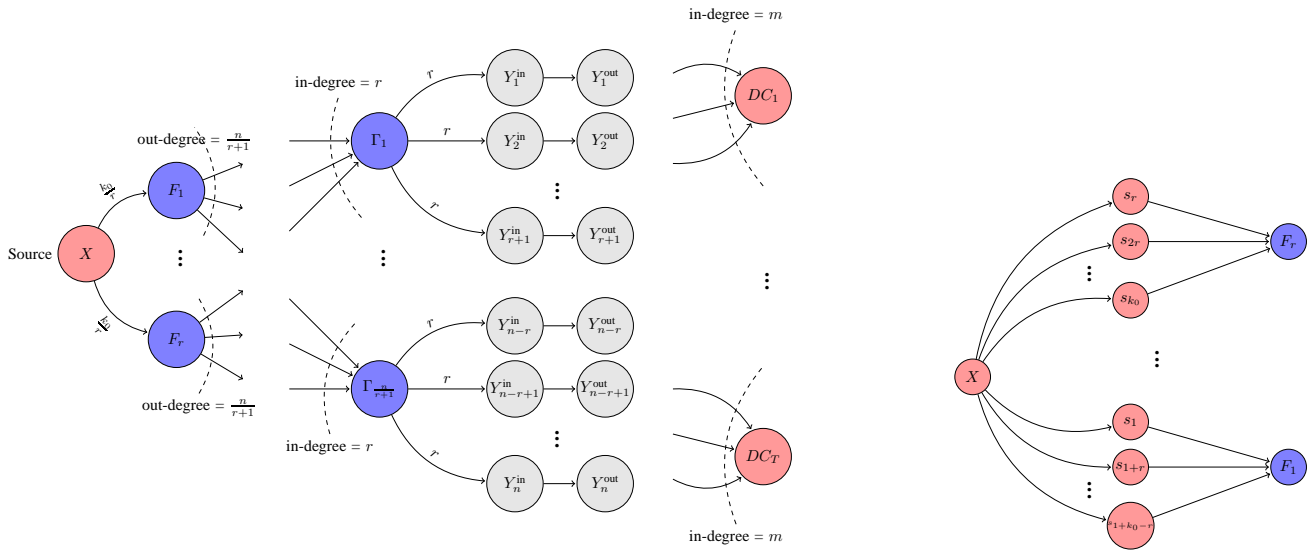
\begin{figure*}[htbp]
\hrulefill
\begin{center}
\scalebox{0.6}{
\begin{tikzpicture}
[every place/.style={minimum size=12.5mm, minimum width=12.5mm}]
\node[ minimum size=55mm, place,draw=black,  fill=red!40, label = left:{Source}] (source) at (-2-3,-2.5) [place] {$X$} ;

\node[ minimum size=12.5mm, place,draw=black,  fill=blue!50] (f1) at (-3,-1.1) [place] {$F_1$} ;
\path  (f1) edge[semithick, post] node [below, midway]  {} (1.5-3,-2.5+2);
\path  (f1) edge[semithick, post] node [below, midway]  {} (1.5-3,-2.5+1); 
\path  (f1) edge[semithick, post] node [below, midway]  {} (1.5-3,-2.5);
\path  (-2.5+.1,-1-1.1) edge[semithick, dashed, bend right] node [below, at end]  {out-degree = $\frac{n}{r+1}$ } (-.1-2.5, 1.4-1.1);
\node[text width=4cm] at (-1,-2.5) {\huge $\vdots$};
\node[ minimum size=12.5mm, place,draw=black,  fill=blue!50] (fr) at (-3,-4.4) [place] {$F_r$} ;
\path  (fr) edge[semithick, post] node [above, midway]  {} (1.5-3,-2.3-3);
\path  (fr) edge[semithick, post] node [above, midway]  {} (1.5-3,-2.3-2);
\path  (fr) edge[semithick, post] node [above, midway]  {} (1.5-3,-2.3-1);
\path  (-2.5,-5.5-.1) edge[semithick, dashed, bend right] node [below, at start]  {out-degree = $\frac{n}{r+1}$ } (-2.5, 2-5.5);
\path (source) edge[semithick, post, bend left] node [midway, sloped, above] {$\frac{k_0}{r}$} (f1);
\path (source) edge[semithick, post, bend right] node [midway, sloped, below] {$\frac{k_0}{r}$} (fr);
\node[ minimum width=15mm, place,draw=black,  fill=blue!50] (gamma1in) at (1.5,0) [place] {$\Gamma_1$} ;
\path  (-.5,-1.5+1.5) edge[semithick, post] node [above, midway]  {} (gamma1in);
\path  (-.5,-1.5+0.5) edge[semithick, post] node [above, midway]  {} (gamma1in);
\path  (-.5,-1.5-0.5) edge[semithick, post] node [above, midway]  {} (gamma1in);
\path  (1.1,-2+.3) edge[semithick, dashed, bend left] node [above, at end]  {in-degree = $r$ } (-0.6+1.2, 1);
\node[ minimum width=15mm, place,draw=black,  fill=black!10] (Y1in) at (3.5+1,1.3+0.1) [place] {$Y_1^{\text{in}}$} ;
\node[ minimum width=15mm, place,draw=black,  fill=black!10] (Y1out) at (5.5+1,1.3+0.1) [place] {$Y_1^{\text{out}}$} ;
\node[ minimum width=15mm, place,draw=black,  fill=black!10] (Y2in) at (3.5+1,0) [place] {$Y_2^{\text{in}}$} ;
\node[ minimum width=15mm, place,draw=black,  fill=black!10] (Y2out) at (5.5+1,0) [place] {$Y_2^{\text{out}}$} ;
\node[ minimum width=15mm, place,draw=black,  fill=black!10] (Yrin) at (3.5+1,-2-0.1) [place] {$Y_{r+1}^{\text{in}}$} ;
\node[ minimum width=15mm, place,draw=black,  fill=black!10] (Yrout) at (5.5+1,-2-0.1) [place] {$Y_{r+1}^{\text{out}}$} ;
\node[text width=4cm] at (7.35,-1) {\huge $\vdots$};
\path (Y1in) edge[semithick, post] node [above, midway]  {} (Y1out);
\path (Y2in) edge[semithick, post] node [above, midway]  {} (Y2out);
\path (Yrin) edge[semithick, post] node [above, midway]  {} (Yrout);
\path (gamma1in) edge[semithick, post, bend left] node [midway, sloped, above] {$r$} (Y1in);
\path (gamma1in) edge[semithick, post] node [midway, above] {$r$} (Y2in);
\path (gamma1in) edge[semithick, post, bend right] node [midway, sloped, above] {$r$} (Yrin);
\node[text width=4cm] at (3.5,-2.5) {\huge $\vdots$};

\node[ minimum width=15mm, place,draw=black,  fill=blue!50] (gammarin) at (1.5,0-5.5) [place] {$\Gamma_{\frac{n}{r+1}}$} ;
\path  (-.5,-6.5+3) edge[semithick, post] node [above, midway]  {} (gammarin);
\path  (-.5,-6.5+2) edge[semithick, post] node [above, midway]  {} (gammarin);
\path  (-.5,-6.5+1) edge[semithick, post] node [above, midway]  {} (gammarin);
\path  (-0.6+1,-1-5.5) edge[semithick, dashed, bend left] node [below, at start]  {in-degree = $r$ } (1, 1.5-5.5	);
\node[ minimum width=15mm, place,draw=black,  fill=black!10] (Yr1in) at (3.5+1,1.3+0.1-5.5) [place] {$Y_{n-r}^{\text{in}}$} ;
\node[ minimum width=15mm, place,draw=black,  fill=black!10] (Yr1out) at (5.5+1,1.3+0.1-5.5) [place] {$Y_{n-r}^{\text{out}}$} ;
\node[ minimum width=15mm, place,draw=black,  fill=black!10] (Yr2in) at (3.5+1,0-5.5) [place] {$Y_{n-r+1}^{\text{in}}$} ;
\node[ minimum width=15mm, place,draw=black,  fill=black!10] (Yr2out) at (5.5+1,0-5.5) [place] {$Y_{n-r+1}^{\text{out}}$} ;
\node[ minimum width=15mm, place,draw=black,  fill=black!10] (Yrrin) at (3.5+1,-2-0.1-5.5) [place] {$Y_n^{\text{in}}$} ;
\node[ minimum width=15mm, place,draw=black,  fill=black!10] (Yrrout) at (5.5+1,-2-0.1-5.5) [place] {$Y_n^{\text{out}}$} ;
\node[text width=4cm] at (7.35,-6.5) {\huge $\vdots$};
\path (Yr1in) edge[semithick, post] node [above, midway]  {} (Yr1out);
\path (Yr2in) edge[semithick, post] node [above, midway]  {} (Yr2out);
\path (Yrrin) edge[semithick, post] node [above, midway]  {} (Yrrout);
\path (gammarin) edge[semithick, post, bend left] node [midway, sloped, above] {$r$} (Yr1in);
\path (gammarin) edge[semithick, post] node [midway, above] {$r$} (Yr2in);
\path (gammarin) edge[semithick, post, bend right] node [midway, sloped, above] {$r$} (Yrrin);
\node[ minimum width=15mm, place,draw=black,  fill=red!40] (DC1) at (9+1,1) [place] {$DC_1$} ;
\node[ minimum width=15mm, place,draw=black,  fill=red!40] (DC4) at (9+1,-7) [place] {$DC_T$} ;
\path  (7+1,1.5) edge[semithick, post, bend left] node [above, midway]  {} (DC1);
\path  (7+1,0.5) edge[semithick, post] node [above, midway]  {} (DC1);
\path  (7+1,-0.5) edge[semithick, post, bend right] node [above, midway]  {} (DC1);
\path  (9+1,-1) edge[semithick, dashed, bend left] node [above, at end]  {in-degree = $m$ } (8.4+1, 2.5	);
\node[text width=4cm] at (12,-3) {\huge $\vdots$};
\path  (7+1,1.5-7) edge[semithick, post, bend left] node [above, midway]  {} (DC4);
\path  (7+1,0.5-7) edge[semithick, post] node [above, midway]  {} (DC4);
\path  (7+1,-0.5-7) edge[semithick, post, bend right] node [above, midway]  {} (DC4);
\path  (8.4+1,-1-7-0.5) edge[semithick, dashed, bend left] node [below, at start]  {in-degree = $m$ } (9+1, 2.5-7-0.5);
\end{tikzpicture}
\begin{tikzpicture}
\node[ minimum width=50mm, place,draw=black,  fill=red!40] (gamma1in) at (3,0) [place] {$X$} ;

\node[ minimum width=80mm, place,draw=black,  fill=red!40] (s11) at (5.5+1,-1) [place] {$s_{1}$} ;
\node[ minimum width=80mm, place,draw=black,  fill=red!40] (s12) at (5.5+1,-2) [place] {$s_{1+r}$} ;
\node[text width=4cm] at (5.5+2,-2.4) {\huge $\vdots$};
\node[ minimum width=80mm, place,draw=black,  fill=red!40] (s1r) at (5.5+1,-3.3) [place] {\tiny $ s_{1+k_0-r}$\normalsize } ;
\node[ minimum width=15mm, place,draw=black,  fill=blue!50] (LD1) at (10,-2) [place] {$F_1$} ;

\node[text width=4cm] at (9,0.5) {\huge $\vdots$};

\node[ minimum width=80mm, place,draw=black,  fill=red!40] (sr1) at (5.5+1,4) [place] {$s_{r}$} ;
\node[ minimum width=80mm, place,draw=black,  fill=red!40] (sr2) at (5.5+1,3) [place] {$s_{2r}$} ;
\node[text width=4cm] at (5.5+2,2.4) {\huge $\vdots$};
\node[ minimum width=80mm, place,draw=black,  fill=red!40] (srr) at (5.5+1,1.7) [place] {$ s_{k_0}$} ;
\node[ minimum width=15mm, place,draw=black,  fill=blue!50] (LD2) at (10,3) [place] {$F_r$} ;

\path (gamma1in) edge[semithick, post,bend right] node [midway, sloped, above] {} (s11);
\path (gamma1in) edge[semithick, post,bend right] node [midway, above] {} (s12);
\path (gamma1in) edge[semithick, post,bend right] node [midway, sloped, above] {} (s1r);
\path (s11) edge[semithick, post] node [above, sloped, midway]  {} (LD1);
\path (s12) edge[semithick, post] node [above, sloped, midway]  {} (LD1);
\path (s1r) edge[semithick, post] node [above, sloped, midway]  {} (LD1);

\path (gamma1in) edge[semithick, post,bend left] node [midway, sloped, above] {} (sr1);
\path (gamma1in) edge[semithick, post,bend left] node [midway, above] {} (sr2);
\path (gamma1in) edge[semithick, post,bend left] node [midway, sloped, above] {} (srr);
\path (sr1) edge[semithick, post] node [above, sloped, midway]  {} (LD2);
\path (sr2) edge[semithick, post] node [above, sloped, midway]  {} (LD2);
\path (srr) edge[semithick, post] node [above, sloped, midway]  {} (LD2);

\end{tikzpicture} 

}
\end{center}

\caption{ 
Left: The  information flow-graph $\mathcal{G}(n,k_0,m,r)$ adapted from \cite{papailiopoulos2012locally}. 
The left-most vertex is the source node $X$.
The  $T={n \choose m}$  vertices $\text{DC}_\mu$ are the destination nodes (referred to as the data collectors).
Each DC is connected to a different  $m$-tuple of $Y_{i}^{\text{out}}$ nodes.
Each of the intermediate nodes $F_\nu, \nu \in [r]$ have out-going edges to all the nodes $\Gamma_\rho, \rho \in \brac*{\frac{n}{r+1}}$.
Right: Equivalent representation for the subgraph containing nodes $F_\nu$ and the source $X$.
}
\label{fig:flow_graph}
\hrulefill
\end{figure*}

We start by describing the graph in \cref{fig:flow_graph} (Left). This graph, $\cG(n,k_0,m,r)$  
consists of a source node $X$ that transmits  $k_0$ $q$-ary symbols to $T={n \choose m}$
data collectors $DC_\mu, \mu \in [T]$. 
We assume that $X$ transmit the secret $\bfs \in \ff_q^{k_0}$.
The unit for the edge capacity is taken to be one $q$-ary 
symbol per channel use. The nodes $F_\nu, \nu \in [r]$ connect to the source $X$ through
 links with capacity $k_0/r$. The edges that connect $\Gamma_\rho, \rho \in \brac{\frac{n}{r+1}}$ to $Y^{\text{in}}_i, i \in [n]$,
 has capacity $r$.
 All the rest of the edges have unit capacity. Each of 
 $\Gamma_\rho, \rho \in \brac{\frac{n}{r+1}}$ have $r$ incoming edges from $F_\nu, \nu \in [r]$.
  The edges $(X,F_\nu)$ are broken into $k_0/r$ unit capacity edges and labelled $s_1, s_2, \ldots, s_{k_0}$ as  shown in the subgraph in \cref{fig:flow_graph} (Right). Node $F_\nu$ connects to the source $X$ through edges $\set{s_{\nu+(\lambda-1)r}}_{\lambda=1}^{k_0/r}, \nu \in [r]$. Let us denote the subset of nodes $\{\Gamma_\rho, \set{Y^{\text{in}}_{(\rho-1)(r+1)+j} }_{j=1}^{r+1}$, $\set{Y^{\text{out}}_{(\rho-1)(r+1)+j} }_{j=1}^{r+1}\}$ as the $\rho^{th}$ repair group. 

A single network use corresponds to a sequence of single data transmission on every edge. Assume that, data transmitted on the edges 
$(Y_i^{\text{in}},Y_i^{\text{out}}), i\in [n]$ in a single network use correspond to the $n$ shares of the secret (i.e., $n$ symbols of $f(\bfs)$, where $f$ is the randomized encoding). Note that, the data collectors connect to $m$ nodes (shares) and obtain all of what $X$ transmits: this must be satisfied for all $m$-subsets (all data collectors). 
We use the network $\cG(n,k_0,m,r)$ to show the existence of a linear $(n,k_0,0,m,r)$-secret sharing scheme.

{\lemma Given that the network $\cG(n,k_0,m,r)$ has multicast capacity  $k_0$, there exists a linear network code with repairability $r$ for this network and the scheme corresponding to the data transmitted on the edges $(Y^{in}_i,Y^{out}_i)$ is an $(n,k_0,0,m,r)$-secret sharing scheme. }

In the following we show that the network $\cG(n,k_0,m,r)$ has multicast capacity $k_9$. 

\begin{definition} A min-cut for any two nodes $v,u$ in  $\cG(n,k_0,m,r)$, denoted $\text{\rm MinCut}(v,u)$,  is defined as a subset of directed edges of minimum aggregate capacity such that if these edges are removed, then there does not exist a path from $v$ to $u$ in the graph $\cG(n,k_0,m,r)$. Let $\abs{\text{\rm MinCut}(v,u)}$ denote the aggregate capacity of the edges in $\text{\rm MinCut}(v,u)$.
\end{definition}

It has been shown \cite{ho2006random,ahlswede2000network} that the minimum of the min-cuts between a single source and multiple sinks corresponds to the 
{\em multicast capacity} of the source. We show that for $\cG(n,k_0,m,r)$ this quantity, $\min_{\mu\in [T]}$ $ \abs{\text{MinCut}(X,DC_\mu)}$, is equal to $k_0$.
\begin{lemma}
 For $\cG(n,k_0,m,r)$ the  multicast capacity is $k_0$. That is,
\begin{equation}\label{mincut}
\min_{\mu\in [T]} \abs{\text{\rm MinCut}(X,DC_\mu)} = k_0.
\end{equation}
\end{lemma}
\begin{IEEEproof}
For $k_0$ satisfying \cref{init_k} we have, 
\begin{equation}\label{m}
m = k_0 + \frac{k_0}{r}-1 = (k_0/r-1) (r+1) + r.
\end{equation}
Suppose that the minimum in \cref{mincut} only contains an $n_1$-subset $\cE$ of edges in $\set*{(X,F_\nu)}_{\nu \in [r]}$. Assume wlog that $\cE = \set{(X,F_1),\ldots,(X,F_{n_1})}$. Consider the data collector $DC_\mu$ that connects to $\gamma_\rho, \rho \in \brac{n/(r+1)}$ nodes in each of the repair groups. If $\gamma_\rho \geq r-n_1$ the min-cut should include all the edges $\set{(F_{n_1+1}, \Gamma_\rho),\ldots,(F_{r},\Gamma_\rho)}$. Otherwise if $\gamma_\rho < r-n_1$ the min-cut includes all the $\gamma_\rho$ edges $(Y^{in}_i,Y^{out}_i)$ in the $\rho^{th}$ repair group connected to $DC_\mu$. Therefore, the minimum in \cref{mincut} would correspond to the data collector that covers entirely as many repair groups as possible. From \cref{m} we see that for a such data collector $\gamma_\rho \geq (r-n_1)$ for all $\rho$ for which $\gamma_\rho>0$ and for all $0\leq n_1 \leq r$. Therefore, 
$$\min_\mu \abs{\text{MinCut}(X,DC_\mu)} =  \frac{k_0}{r} (r-n_1) + n_1 \frac{k_0}{r} = k_0$$
\end{IEEEproof}

We know therefore that a random LNC achieves the multicast capacity $k_0$ for this network. This random LNC corresponds to a secret-sharing scheme with $n$
shares such that the secret in $\ff_q^{k_0}$ can be recovered by looking at any $m$ shares.
 Now to satisfy the local repairability constraint for this LNC, consider the subgraph containing the nodes in the $\rho^{th}$ repair group. Another set of local decoding requirements are imposed on this subgraph. For each $r$-subset of nodes in any local repair group, a local data collector $LD_i, i \in [n]$ connecting to these nodes should be able to decode the input to $\Gamma_\rho$. 
 There are in total $n$ such local decoding requirements.
   These decoding requirements are similar to the local repairability requirements for the network flow graph considered in \cite{papailiopoulos2012locally}. Let $\bfz_\rho \in \bbF_q^r$ denote the data received by $\Gamma_\rho$. Let $N_i$ denote the $r\times r$ local encoding matrix, for the edges $\set{ (\Gamma_\rho, Y^{in}_{{(\rho-1)(r+1)+j}}) }_{j \in [r+1]\setminus\set{i}}$ corresponding to $i^{th}$ local data collector.  Therefore, the data received by the $i^{th}$ local decoder is,
\begin{equation}\label{repairability_constr}
\bfz_\rho N_i, i\in \set*{(\rho-1)(r+1)+1, \ldots, \rho(r+1)}
\end{equation}
We see that, for any local data collector $LD_i$ to recover the data from the node $\Gamma_\rho$ matrix $N_i$ must be full rank. Since we know that for a large enough alphabet size $q$ we can satisfy these constraints \cite[lemma 4]{papailiopoulos2012locally}, there must exist an 
LNC that satisfies the local repair requirements. Therefore, we can construct an $(n,k_0,0,m,r)$-secret-sharing scheme. 


Suppose we write the secret as $\bfs = (s_1, \dots, s_{k_0})$, and term $s_1, \ldots, s_{k_0}$ as the information	symbols.
Now, for the random LNC obtained above that satisfy
the repairability and recovery requirements, we relabel $k = k_0-\ell$
 information symbols $\set{s_{\ell+1}, \ldots, s_{k_0}}$ from the source $X$ as 
secure information symbols and the choose each of the rest $\ell$ 
symbols $\set{s_1, \ldots, s_{\ell}}$ according to a uniformly random distribution 
in $\ff_q$. For such a  random LNC to be $\ell$-secure any eavesdropper 
$ED_\tau, \tau \in [{n \choose \ell}]$ connecting to any
 $\ell$ nodes $Y^{out}_i$ may
be able to recover at most the redundant $\ell$ symbols $\set{s_1, \ldots, s_{\ell}}$
and should have full ambiguity about $\set{s_{\ell+1}, \ldots, s_{k_0}}$.
 We show that these additional security constraints can be satisfied for a random LNC with
large enough alphabet and hence we have an $(n,k,\ell,m,r)$-secret-sharing scheme  satisfying \cref{eq:trivial}. 

Note that if a code is secure against an eavesdropper who can observe any of the $\ell$ shares,
 it must be secure against any adversary who can only observe less than $\ell$ shares.
 Therefore, for $\ell>r$ we can ignore all eavesdroppers who choose all the $(r+1)$ shares 
 of the same	repair group. Since one of the shares
  in a repair group can be recovered from the other $r$ shares,
   an eavesdropper who reads $t$ entire repair groups  is observing
  effectively  only $\ell-t$ shares.
 Therefore, we only need to consider the eavesdroppers that observe a maximum
of $r$ shares in a repair group. Let us denote this
 sub-set of eavesdropper as $ED_\tau, \tau \in \cW^\prime, \cW^\prime \subseteq [{n \choose \ell}]$.

If $(c_1, \dots , c_n)$ are the $n$ shares
for the secret $\bfs$, we must have
the data transmitted on the edges $(Y^{in}_i,Y^{out}_i)$ with the following linear form,
\begin{equation}\label{coefF_matrix}
 \begin{pmatrix}
  c_1 \\
  \vdots \\
  c_{n}
 \end{pmatrix}	=
 \begin{pmatrix}
  a_{1,1} & a_{1,2} & \cdots & a_{1,{k_0}} \\
  a_{2,1} & a_{2,2} & \cdots & a_{2,{k_0}} \\
  \vdots  & \vdots  & \ddots & \vdots  \\
  a_{n,1} & a_{n,2} & \cdots & a_{n,{k_0}}
 \end{pmatrix}	
 \begin{pmatrix}
  s_1 \\
  \vdots \\
  s_{k_0}
 \end{pmatrix} = A \bf{s}.
\end{equation}
We claim that the security against an eavesdropper $ED_\tau, \tau \in \cW^\prime$ 
is equivalent to a full-rank requirement on a $\ell\times \ell$ sub-matrix of $A$.

\begin{lemma}\label{independence_lemma}
Let $\cE^\tau = \set{e^\tau_1, e^\tau_2, \ldots, e^\tau_{\ell}} \subseteq [n]$ denotes  the shares  an eavesdropper $ED_\tau$ can observe. We have,
\begin{align}
 \bfc_{\cE^\tau}	= A^\tau_1 \bfs_{[\ell]}+ A^\tau_2 \bfs_{[k_0]\setminus [\ell]}.
 \end{align}
 If for all eavesdroppers $ED_\tau, \tau \in \cW^\prime$ the $\ell \times \ell$ matrix $A^\tau_1$ is full-rank then the LNC  is $\ell$-secure.
\end{lemma}
\begin{IEEEproof}
Suppose for some specific $\tau \in \cW'$,
{\footnotesize $$
 A^\tau_1 =  
 \begin{pmatrix}
  a_{{e_1} ,1} & a_{{e_1} ,2} & \cdots & a_{{e_1} ,{\ell}} \\
  a_{{e_2} ,1} & a_{{e_2} ,2} & \cdots & a_{{e_2} ,{\ell}} \\
  \vdots  & \vdots  & \ddots & \vdots  \\
  a_{{e_{\ell}} ,1} & a_{{e_{\ell}} ,2} & \cdots & a_{{e_{\ell}} ,{\ell}}
 \end{pmatrix};
 A^\tau_2 =
 \begin{pmatrix}
  a_{{e_1} ,\ell+1} & \cdots & a_{{e_1} ,{k_0}} \\
  a_{{e_2} ,\ell+1} & \cdots & a_{{e_2} ,{k_0}} \\
  \vdots  & \vdots  & \ddots   \\
  a_{{e_{\ell}} ,\ell+1} & \cdots & a_{{e_{\ell}} ,{k_0}}
 \end{pmatrix}.	
$$\small}
Since $A^\tau_1$ is full rank, there must be a unique solution to $s_{1}, s_{2}, \ldots ,s_{\ell}$  for every value	of $\bfc_{\cE^\tau}$ and every value of $\set{s_{\ell+1}, \ldots , s_{k_0}}\in\ff_q^{k_0}$. 
Hence, we have,
$$H(\bfs_{[\ell]}|\bfc_{\cE^\tau},\bfs_{[k_0]\setminus [\ell]})=0$$
We therefore have the following chain of inequalities that establishes that the eavesdropper does not get any information about the secret 
 from his observation. 
 
$
I(\bfs_{[k_0]\setminus [\ell]};\bfc_{\cE^\tau}) = H(\bfc_{\cE^\tau}) - H(\bfc_{\cE^\tau} |\bfs_{[k_0]\setminus [\ell]})
\leq \ell - H(\bfc_{\cE^\tau} |\bfs_{[k_0]\setminus [\ell]}) + H(\bfc_{\cE^\tau} |\bfs_{[\ell]},\bfs_{[k_0]\setminus [\ell]})
= \ell - I(\bfc_{\cE^\tau},\bfs_{[\ell]} |\bfs_{[k_0]\setminus [\ell]}) 
= \ell - H(\bfs_{[\ell]}| \bfs_{[k_0]\setminus [\ell]}) + H(\bfs_{[\ell]}|\bfc_{\cE^\tau},\bfs_{[k_0]\setminus [\ell]})
= \ell - H(\bfs_{[\ell]})
= \ell - \ell
= 0.
$
\end{IEEEproof}


We also have the following lemma.
\begin{lemma} \label{mincut_e} Consider the subgraph $\cG_e$
 formed by removing the edges $s_{\ell+1}, \ldots, s_{k_0}$ from the graph $\cG(n,k_0,m,r)$. For this modified
  network graph the multicast capacity between the source and the eavesdroppers $ED_\tau , \tau \in \cW^\prime$ is $\ell$ i.e.
$$\min_{\tau \in \cW^\prime} \abs{\text{\rm MinCut}(X,ED_\tau)} = \ell.$$
\end{lemma}
\begin{IEEEproof}
It is easy to see from the network structure that min-cut for every eavesdropper $ED_\tau, \tau \in \cW^\prime$ corresponds to all the edges $(Y^{in}_i, Y^{out}_i)$ to which an eavesdropper connects in each repair group. Since, every eavesdropper in $\cW^\prime$ connects to $\ell$ nodes, the minimum mincut is also $\ell$.
\end{IEEEproof}
Consider the eavesdropper $ED_\tau, \tau\in \cW^\prime$ which connects to $t_1, t_2, \ldots, t_{n/(r+1)}$ nodes in each of the repair groups. Therefore, we have
$$\sum_{\rho=1}^{n/(r+1)} t_\rho = \ell$$
where  $0\leq t_\rho \leq r, \forall \rho \in [n/(r+1)]$. Let $N_\rho^\prime, \rho \in [n/(r+1)]$ 
denote the $t_\rho \times r$ local encoding sub-matrix of 
$N_\rho$ (see, \cref{repairability_constr}) for the edges $(\Gamma_\rho, Y^{in}_i)$ 
connecting the eavesdropper to the $\rho^{th}$ repair group. Also, let
 $D_\rho, \rho\in [n/(r+1)]$ denote the $r \times \ell$ matrix corresponding to 
 the local encoding vectors  for $(F_\nu,\Gamma_\rho), \nu \in [r] $,  
for the induced graph $\cG_e$ described above. 
The matrix $A^\tau_1$ from  \cref{independence_lemma}  can be written as,
\begin{align}
A^\tau_1=
 \begin{pmatrix}
  N_1^\prime D_1 \\ N_2^\prime D_2 \\ \vdots \\ N_{\frac{n}{r+1}}^\prime D_{\frac{n}{r+1}} 
 \end{pmatrix}.
\end{align}
We need all of the matrices $A^\tau_1, \tau \in \cW'$ to be full-rank simultaneously. 
 Now  using \cref{mincut_e} we can see that these  constraints on the matrices $\bf{D}_\rho$s can all be satisfied simultaneously --with the local repairability and multicast capacity--  for all $\tau \in \cW^\prime$ for a large enough alphabet size \cite{ho2006random}, \cite[Lemma 4]{dimakis2010network}. Therefore, a random LNC  satisfies the full rank constraints of \cref{independence_lemma}. 

Therefore, for the random LNC obtained above, for any eavesdropper $ED_\tau$ observing $\cE^\tau \subseteq [n]$,
$I(\bfs_{[k_0]\setminus [\ell]};\bfc_{\cE^\tau})=0.$
Since the data collectors can recover $\bfs$ from any $m$ nodes and $H(\bfs_{[k_0]\setminus [\ell]}|\bfs) = 0$, the secret
 is recoverable from any $m$ shares. Therefore, we have an $(n,k,\ell,m,r)$-scheme achieving the upper bound in \cref{eq:trivial}.


\end{appendices}

\bibliographystyle{abbrv}
\bibliography{aryabib}

\balance

\end{document}